%% file: main.tex
\pdfoutput=1
\newif\ifdraft \draftfalse \newif\iffull \fulltrue

\newcommand{\citet}{\cite}
\newcommand{\citep}{\cite}

\newcommand{\myheading}[1]{\paragraph*{\mbox{\textbf{#1}}}}
\documentclass{llncs}
\input{prelude}

\begin{document} 
\title{An Assertion-Based Program Logic \\ for Probabilistic
Programs\thanks{This is the \iffull full \else conference \fi version of the paper.}}

\author{Gilles Barthe \and Thomas Espitau \and Marco Gaboardi \\
  Benjamin Gr\'egoire \and Justin Hsu \and Pierre-Yves Strub}

\institute {IMDEA Software Institute, Spain
  \and Université Paris 6, France
  \and University at Buffalo, SUNY, USA
  \and Inria Sophia Antipolis--Méditerranée, France
  \and University College London, UK
  \and École Polytechnique, France}

\maketitle
\begin{abstract}
  \iffull
Research on deductive verification of probabilistic programs has considered
expectation-based logics, where pre- and post-conditions are
real-valued functions on states, and assertion-based logics, where pre- and
post-conditions are boolean predicates on state distributions. Both approaches
have developed over nearly four decades, but they have different standings
today.  Expectation-based systems have managed to formalize many sophisticated
case studies, while assertion-based systems today have more limited
expressivity and have targeted simpler examples.
\fi

We present \SYSTEM, a sound and relatively complete assertion-based
program logic, and demonstrate its expressivity by verifying several
classical examples of randomized algorithms using an implementation
in the \EasyCrypt proof assistant. \SYSTEM features new proof rules for
loops and adversarial code, and supports richer assertions than existing program logics.
We also show that \SYSTEM allows convenient reasoning about complex
probabilistic concepts by developing a new program logic for probabilistic
independence and distribution law, and then smoothly
embedding it into \SYSTEM.
\iffull
Our work demonstrates that the
assertion-based approach is not fundamentally limited and suggests
that some notions are potentially easier to reason about in
assertion-based systems.
\fi
\end{abstract}

\section{Introduction}
The most mature systems for deductive verification of randomized algorithms are
\emph{expectation-based} techniques; seminal examples include \Sppdl
\citep{Kozen85} and \Spgcl \citep{Morgan96}.  These approaches reason about
\emph{expectations}, functions $E$ from states to real numbers,\footnote{%
  Treating a program as a function from input states $s$ to output distributions
$\mu(s)$, the expected value of $E$ on $\mu(s)$ is an expectation.}
propagating them backwards through a program until they are transformed into a
mathematical function of the input. Expectation-based systems are both
theoretically elegant
\citep{KaminskiKMO16,gretz2014operational,olmedo2016reasoning,kaminski2016inferring}
and practically useful; implementations have verified numerous randomized
algorithms \citet{Hurd03,HurdMM05}.  However, properties involving
multiple probabilities or expected values can be cumbersome to verify---each
expectation must be analyzed separately.

An alternative approach envisioned by Ramshaw \citet{Ramshaw79} is to work with
predicates over distributions. A direct
comparison with expectation-based techniques is difficult, as the approaches are quite
different. In broad strokes, assertion-based systems can verify richer
properties in one shot and have specifications that are arguably more
intuitive, especially for reasoning about loops, while
expectation-based approaches can transform expectations mechanically
and can reason about non-determinism.  However, the comparison is not
very meaningful for an even simpler reason: existing assertion-based
systems such as \citep{ChadhaCMS07,Hartog:thesis,RandZ15} are not as
well developed as their expectation-based counterparts.
\begin{description}
  \item[Restrictive assertions.]
    Existing probabilistic program logics do not support reasoning about
    expected values, only probabilities. As a result, many properties about
    average-case behavior are not even expressible.
  \item[Inconvenient reasoning for loops.]
    The Hoare logic rule for
    deterministic loops does not directly generalize to probabilistic
    programs. Existing assertion-based systems either forbid
    loops, or impose complex semantic side conditions to
    control which assertions can be used as loop invariants. Such side
    conditions are restrictive and difficult to establish.
  \item[No support for external or adversarial code.]
    A strength of expectation-based techniques is reasoning
    about programs combining probabilities and
    \emph{non-determinism}. In contrast, Morgan and
    McIver~\cite{McIverM05} argue that assertion-based techniques
    cannot support compositional reasoning for such a combination. For
    many applications, including cryptography, we would still like to
    reason about a commonly-encountered special case: programs using
    external or adversarial code. Many security properties in
    cryptography boil down to analyzing such programs, but existing
    program logics do not support adversarial code.
  \item[Few concrete implementations.]
    There are by now several
    independent implementations of expectation-based techniques,
    capable of verifying interesting probabilistic programs. In
    contrast, there are only scattered implementations of
    probabilistic program logics.
\end{description}
These limitations raise two points. Compared to expectation-based
approaches:
\begin{enumerate}
\item Can assertion-based approaches achieve similar expressivity?
\item Are there situations where assertion-based approaches are more suitable?
\end{enumerate}
In this paper, we give positive evidence for both of these
points.\footnote{Note that we do not give mathematically precise
  formulations of these points; as we are interested in the practical
  verification of probabilistic programs, a purely theoretical answer
  would not address our concerns.}
Towards the first point, we give a
new assertion-based logic \SYSTEM for probabilistic programs,
overcoming limitations in existing probabilistic program
logics. \SYSTEM supports a rich set of assertions that can express
concepts like expected values and probabilistic independence, and
novel proof rules for verifying loops and adversarial code.  We prove
that \SYSTEM is sound and relatively complete.

Towards the second point, we evaluate \SYSTEM in two ways. First, we
define a new logic for proving probabilistic independence and
distribution law properties---which are difficult to capture with
expectation-based approaches---and then embed it into \SYSTEM. This
sub-logic is more narrowly focused than \SYSTEM, but supports more
concise reasoning for the target assertions. Our embedding
demonstrates that the assertion-based approach can be flexibly
integrated with intuitive, special-purpose reasoning principles.  To
further support this claim, we also provide an embedding of the Union
Bound logic, a program logic for reasoning about accuracy
bounds~\cite{BartheGGHS16-icalp}. Then, we
develop a full-featured implementation of \SYSTEM in the \EasyCrypt
theorem prover and exercise the logic by mechanically verifying a
series of complex randomized algorithms. Our results suggest that the
assertion-based approach can indeed be practically viable.

\myheading{Abstract logic.}
To ease the presentation, we present \SYSTEM in two stages.  First, we
consider an abstract version of the logic where assertions are general
predicates over distributions, with no compact syntax. Our abstract
logic makes two contributions: reasoning for loops, and for
adversarial code.

\paragraph*{Reasoning about Loops.}
Proving a property of a probabilistic loop typically requires establishing a
loop invariant, but the class of loop invariants that can be soundly used
depends on the termination behavior---stronger termination assumptions allows
richer loop invariants. We identify three classes of assertions
that can be used for reasoning about probabilistic loops, and provide a proof
rule for each one:
\begin{itemize}
\item arbitrary assertions for \emph{certainly terminating} loops,
  i.e.\ loops that terminate in a finite amount of iterations;
\item \emph{topologically closed} assertions for \emph{almost surely}
  terminating loops, i.e.\ loops terminating with probability $1$;
\item \emph{downwards closed} assertions for arbitrary loops.
\end{itemize}
The definition of topologically closed assertion is reminiscent of Ramshaw
\citet{Ramshaw79}; the stronger notion of downwards closed assertion appears to
be new.

Besides broadening the class of loops that can be analyzed, our rules
often enable simpler proofs. For instance, if the loop is certainly
terminating, then there is no need to prove semantic side-conditions.
Likewise, there is no need to consider the termination behavior of the
loop when the invariant is downwards and topologically closed. For
example, in many applications in cryptography, the target property is
that a \lq\lq bad\rq\rq\ event has low probability: $\Pr{[E]} \leq
k$. In our framework this assertion is downwards and topologically
closed, so it can be a loop invariant regardless of the termination
behavior.

\paragraph*{Reasoning about Adversaries.}
Existing assertion-based logics cannot reason about probabilistic programs with
\emph{adversarial} code.
\emph{Adversaries} are special probabilistic procedures consisting of an interface
listing the concrete procedures that an adversary can call (\emph{oracles}),
along with restrictions like how many calls an
adversary may make. Adversaries are useful in cryptography, where
security notions are described using experiments in which
adversaries interact with a challenger, and in game theory and
mechanism design, where adversaries can represent strategic
agents. Adversaries can also model inputs to \emph{online} algorithms.

We provide proof rules for reasoning about adversary calls. Our rules
are significantly more general than previously considered rules for
reasoning about adversaries. For instance, the rule for adversary used
by~\citet{BartheGGHS16-icalp} is restricted to adversaries that cannot
make oracle calls.

\paragraph*{Metatheory.}
We show soundness and relative completeness of the core abstract logic, with
mechanized proofs in the \textsc{Coq} proof assistant.\footnote{%
The formalization is available at \url{https://github.com/strub/xhl}.}

\myheading{Concrete logic.}

While the abstract logic is conceptually clean, it is inconvenient
for practical formal verification---the
assertions are too general and the rules involve
semantic side-conditions. To address these issues,
we flesh out a concrete version of \SYSTEM.
% with a grammar for
% probabilistic assertions, sufficient syntactic conditions for
% the semantic side-conditions, and automated tools for transforming the
% assertions syntactically.
Assertions are described by a grammar modeling a two-level assertion
language. The first level contains state predicates---deterministic assertions
about a single memory---while the second layer contains probabilistic predicates
constructed from probabilities and expected values over discrete distributions.
While the concrete assertions are theoretically less expressive than their
counterparts in the abstract logic, they can already encode
common properties and notions from existing proofs, like probabilities,
expected values, distribution laws and probabilistic independence. Our
assertions can express theorems from probability theory, enabling
sophisticated reasoning about probabilistic concepts.

Furthermore, we leverage the concrete syntax to simplify verification.
\begin{itemize}
  \item We develop an automated procedure for generating pre-conditions of
    non-looping commands, inspired by expectation-based systems.
  \item We give syntactic conditions for the closedness and termination
    properties required for soundness of the loop rules.
\end{itemize}

\myheading{Implementation and case studies.}  We implement \SYSTEM on
top of \EasyCrypt, a general-purpose proof assistant for reasoning
about probabilistic programs, and we mechanically verify a diverse
collection of examples including textbook algorithms and a randomized
routing procedure. We develop an \EasyCrypt formalization of
probability theory from the ground up, including tools like
concentration bounds (e.g., the Chernoff bound), Markov's inequality,
and theorems about probabilistic independence.

\myheading{Embeddings.}
We propose a simple program logic for proving \emph{probabilistic independence}.  This
logic is designed to reason about independence in a lightweight way,
as is common in paper proofs. We prove that the logic can be embedded
into \SYSTEM, and is therefore sound.
%% To our knowledge, there is no principled method for reasoning about
%% probabilistic independence in systems like \Spgcl. The challenge is
%% that probabilistic independence involves multiple events: two events
%% $E$ and $F$, and their conjunction $E \land F$. Assertion-based
%% approaches seem more promising for proving these complex properties.
Furthermore, we prove an embedding of the Union Bound
logic~\cite{BartheGGHS16-icalp}.

% ------------------------------------------------------------------------

\section{Mathematical Preliminaries}

As is standard, we will model randomized computations using \emph{sub-distributions}.

\begin{definition}
  A \emph{sub-distribution} over a set $A$ is defined by a mass function
  $\mu : A \to [0,1]$ that gives the probability of the unitary events $a \in
  A$. This mass function must be s.t.  $\sum_{a \in A} \mu(a)$ is well-defined
  and
  $\wt{\mu} \eqdef \sum_{a\in A} \mu(a) \leq 1$.
  In particular, the \emph{support}
  $\supp(\mu) \eqdef \{ a \in A \mid \mu(a) \neq 0 \}$
  is discrete.\footnote{%
    We work with discrete distributions to keep
    measure-theoretic technicalities to a minimum, though we do not see
  obstacles to generalizing to the continuous setting.}
  The name ``sub-distribution'' emphasizes that the total probability may be
  strictly less than $1$.
  When the \emph{weight} $\wt{\mu}$ is equal to $1$, we call $\mu$ a
  \emph{distribution}.  We let $\Dist(A)$ denote the set of
  sub-distributions over $A$.
  The probability of an event $E(x)$ w.r.t. a sub-distribution $\mu$,
  written $\Pr_{x \sim \mu} [E(x)]$, is defined as
  $\sum_{x \in A \mid E(x)} \mu(x)$.
\end{definition}

Simple examples of sub-distributions include the \emph{null sub-distribution}
$\mathbf{0}$, which maps each element of the underlying space to $0$; and the
\emph{Dirac distribution centered on $x$}, written $\dunit{x}$, which maps $x$
to $1$ and all other elements to $0$. The following standard
construction gives a monadic structure to sub-distributions.

\begin{definition}
  Let $\mu \in \Dist(A)$ and $f : A \to \Dist(B)$. Then
  $\dlet {a} {\mu} {f} \in \Dist(B)$ is defined by
  \[
    \dlet{a} {\mu} {f} (b) \eqdef \sum_{a \in A} \mu(a) \cdot f(a)(b) .
  \]
  We use notation reminiscent of expected values, as the definition is quite
  similar.
\end{definition}

We will need two constructions to model branching statements.

\begin{definition}
Let $\mu_1,\mu_2\in\Dist(A)$ such that $\wt{\mu_1}+\wt{\mu_2}\leq
1$. Then $\mu_1+\mu_2$ is the sub-distribution $\mu$ such that
$\mu(a)=\mu_1(a)+\mu_2(a)$ for every $a\in A$.
\end{definition}
\begin{definition}
Let $E \subseteq A$ and $\mu \in \Dist(A)$. Then the restriction
$\drestr \mu E$ of $\mu$ to $E$ is the sub-distribution
such that $\drestr \mu E (a)= \mu(a)$ if $a\in E$ and 0 otherwise. 
\end{definition}

Sub-distributions are partially ordered under the pointwise order.

\begin{definition}
  Let $\mu_1,\mu_2\in\Dist(A)$. We say $\mu_1\leq \mu_2$ if $\mu_1(a) \leq
  \mu_2(a)$ for every $a\in A$, and we say $\mu_1 = \mu_2$ if
  $\mu_1(a) = \mu_2(a)$ for every $a\in A$.
\end{definition}
We use the following lemma when reasoning about the semantics of loops.
\begin{lemma}\label{lem:less:eq:distr}
  If $\mu_1\leq\mu_2$ and $\wt{\mu_1}=1$, then $\mu_1=\mu_2$ and $\wt{\mu_2}=1$.
\end{lemma}

Sub-distributions are stable under pointwise-limits.

\begin{definition}
  A sequence $(\mu_n)_{n\in\NN} \in\Dist(A)$ sub-distributions
  \emph{converges} if for every $a \in A$, the sequence
  $(\mu_n(a))_{n\in\NN}$ of real numbers converges. The \emph{limit
  sub-distribution} is defined as
  \[
    \mu_\infty(a) \eqdef \lim_{n \to \infty} \mu_n(a)
  \]
  for every $a \in A$.  We write $\lim_{n\rightarrow\infty} \mu_n$ for $\mu_\infty$.
\end{definition}

\begin{lemma}
  \label{lem:limitProb}
  Let $(\mu_n)_{n\in\NN}$ be a convergent sequence of sub-distributions. Then
  for any event $E(x)$, we have:
  \[
    \forall n\in \NN.\,  \Pr_{x \sim 
    \mu_\infty} [E(x)] = \lim_{n \to \infty} \Pr_{x \sim \mu_n} [E(x)].
  \]
\end{lemma}
Any bounded increasing real sequence has a limit; the same is true of
sub-distributions.
\begin{lemma}\label{lem:lim:distr}
  Let $(\mu_n)_{n\in\NN} \in\Dist(A)$ be an increasing sequence of
  sub-distributions. Then, this sequence converges to $\mu_\infty$
  and $\mu_n \leq \mu_\infty$ for every $n\in\NN$. In
  particular, for any event $E$, we have $\Pr_{x \sim \mu_n} [E]\leq
  \Pr_{x \sim \mu_\infty} [E]$ for every $n\in \NN$.
\end{lemma}

\section{Programs and Assertions}\label{sec:programs}

Now, we introduce our core programming language and its denotational semantics.

\paragraph*{Programs.}
We base our development on \pWhile, a strongly-typed imperative
language with deterministic assignments, probabilistic assignments,
conditionals, loops, and an $\abort$ statement which halts the
computation with no result. Probabilistic assignments
$x \rnd g$ assign a value sampled from a distribution $g$ to a program
variable $x$. The syntax of statements is defined by the grammar:
\begin{align*}
  s &::= \skp
           \mid \abort
           \mid x \asn e
           \mid x \rnd g 
           \mid s; s \\
           &\mid \ifstmt{e}{s}{s}
           \mid \while{e}{s}
           \mid \call{x}{\IProc}{e}
           \mid \call{x}{\AProc}{e}
\end{align*}
where $x$, $e$, and $g$ range over typed variables in $\Var$,
expressions in $\expr$ and distribution expressions in $\dexpr$
respectively. The set $\expr$ of well-typed expressions is defined
inductively from $\Var$ and a set $\mathcal{F}$ of
function symbols, while the set $\dexpr$ of well-typed distribution
expressions is defined by combining a set of distribution symbols
$\mathcal{S}$ with expressions in $\mathcal{E}$. Programs may call a set
$\IProc$ of internal procedures as well as a set $\AProc$ of external
procedures. We assume that we have code for internal procedures but
not for external procedures---we only know indirect information,
like which internal procedures they may call.
Borrowing a convention from cryptography, we call internal
procedures \emph{oracles} and external procedures \emph{adversaries}.

%------------------------------------------------------------------------------
\paragraph*{Semantics.}

\begin{figure*}[t]
\begin{align*}
  \dsem{m}{\skp} &= \dunit{m} \\
  \dsem{m}{\abort} &= \mathbf{0} \\
  \dsem{m}{x \asn e} &= \dunit{m\subst{x}{\dsem{m}{e}}} \\
  \dsem{m}{x \rnd g} &= \dlet v {\dsem{m}{g}} {\dunit{m[x:=v]}} \\
  \dsem{m}{s_1; s_2} &= \dlet {m'} {\dsem{m}{s_1}} {\dsem{m'}{s_2}}  \\
  \dsem{m}{\ifte{e}{s_1}{s_2}} &=
    \text{if $\dsem{m}{e}$ then $\dsem{m}{s_1}$ else $\dsem{m}{s_2}$} \\
  \dsem{m}{\while{e}{s}} &=
    \lim_{n \to \infty}\ \dsem{m}{(\ift e s)^n;\ift e \abort} \\
  \dsem{m}{\call{x}{\IProc}{e}} &=
    \dsem{m}{\farg{f} \asn e; \fbody{f}; x \asn \fret{f}} \\
  \dsem{m}{\call{x}{\AProc}{e}} &=
    \dsem{m}{\farg{a} \asn e; \fbody{a}; x \asn \fret{a}}
    \\ \\ \hline \\
  \dsem{\state}{s} &= \dlet m \state {\dsem{m}{s}}
\end{align*} 

\caption{\label{fig:semantics} Denotational semantics of programs}
\end{figure*}

The denotational semantics of programs is adapted from the seminal
work of \citet{Kozen79} and interprets programs as sub-distribution
transformers. We view states as type-preserving mappings from
variables to values; we write $\Mem$ for the set of states and
$\Dist(\Mem)$ for the set of probabilistic states. For each procedure
name $f \in \IProc \cup \AProc$, we assume a set
$\VarL[f] \subseteq \Var$ of \emph{local variables} s.t.
$\VarL[f]$ are pairwise disjoint. The other variables
$\Var \setminus \bigcup_f \VarL[f]$ are \emph{global variables}.

To define the interpretation of expressions and distribution
expressions, we let $\dsem{m}{e}$ denote the interpretation of
expression $e$ with respect to state $m$, and $\dsem{\state}{e}$
denote the interpretation of expression $e$ with respect to an initial
sub-distribution $\state$ over states defined by the clause
$\dsem{\state}{e}\eqdef \dlet m \state {\dsem{m}{e}}$. Likewise,
we define the semantics of commands in two stages: first interpreted
in a single input memory, then interpreted in an input
sub-distribution over memories.
\begin{definition}
  The semantics of commands are given in \cref{fig:semantics}.
  \begin{itemize}
  \item The semantics $\dsem{m}{s}$ of a statement $s$ in initial state $m$ is a
    sub-distribution over states.
  \item The (lifted) semantics $\dsem{\state}{s}$ of a statement $s$
    in initial sub-distribution $\state$ over states is a
    sub-distribution over states.
  \end{itemize}
\end{definition}
We briefly comment on loops. The semantics of a loop $\while{e}{c}$ is defined
as the limit of its lower approximations, where the $n$-th \emph{lower
approximation} of $\dsem{\state}{\while{e}{c}}$ is $\dsem{\state}{(\ift e
s)^n;\ift e \abort}$, where $\ift{e}{s}$ is shorthand for $\ifte{e}{s}{\skp}$
and $c^n$ is the $n$-fold composition $c;\cdots;c$.  Since the sequence is
increasing, the limit is well-defined by \cref{lem:lim:distr}.  In contrast, the
$n$-th \emph{approximation} of $\dsem{\state}{\while{e}{c}}$ defined by
$\dsem{\state}{(\ift e s)^n}$ may not converge, since they are not necessarily
increasing.  However, in the special case where the output distribution has
weight $1$, the $n$-th lower approximations and the $n$-th approximations have
the same limit.
\begin{lemma}
  If the sub-distribution $\dsem{\state}{\while{e}{c}}$ has weight $1$, then the
  limit of $\dsem{\state}{(\ift e s)^n}$ is defined and
  \[
    \lim_{n \to \infty}\ \dsem{\state}{(\ift e s)^n;\ift e \abort}
    = \lim_{n \to \infty}\ \dsem{\state}{(\ift e s)^n} .
  \]
\end{lemma}
This follows by \cref{lem:less:eq:distr}, since lower approximations are below
approximations so the limit of their weights (and the weight of their
limit) is $1$. It will be useful to identify programs that terminate with
probability $1$.

\begin{definition}[Lossless]
A statement $s$ is \emph{lossless} if for every
sub-distribution $\state$, $\wt{\dsem{\state}{s}} =\wt{\state}$, where $|\state|$
is the total probability of $\state$. Programs that are not lossless are called \emph{lossy}.
\end{definition}

Informally, a program is lossless if all probabilistic assignments sample from
full distributions rather than sub-distributions, there are no $\abort$
instructions, and the program is almost surely terminating, i.e.\ infinite
traces have probability zero.  Note that if we restrict the language to sample
from full distributions, then losslessness coincides with almost sure
termination.

Another important class of loops are loops with a uniform upper bound
on the number of iterations. Formally, we say that a loop
$\while{e}{s}$ is \emph{certainly terminating} if there exists $k$
such that for every sub-distribution $\state$, we have $
\wt{\dsem{\state}{\while{e}{s}}} =
\wt{\dsem{\state}{(\ift{e}{s})^k}}$. Note that certain termination of
a loop does not entail losslessness---the output distribution of the
loop may not have weight $1$, for instance, if the loop samples from a
sub-distribution or if the loop aborts with positive probability.

\paragraph*{Semantics of Procedure Calls and Adversaries.}
The semantics of internal procedure calls is
straightforward. Associated to each procedure name $f \in \IProc$, we
assume a designated input variable $\farg{f} \in \VarL[f]$, a
piece of code $\fbody{f}$ that executes the function call, and a
result expression $\fret{f}$. A function call $\call{x}{\IProc}{e}$ is
then equivalent to $\farg{f} \asn e; \fbody{f}; x \asn \fret{f}$.
Procedures are subject to well-formedness criteria: procedures should
only use local variables in their scope and after initializing them,
and should not perform recursive calls.
% Well-formedness can be enforced with a simple proof system,
% that takes as input a well-founded order on procedure names, with the
% main procedure as its top elements and checks (among other things)
% that calls are performed on smaller procedures with respect to this
% order.

\jh{The role of $\fbody{a}$ here is also confusing. I put
  ``unspecified'', but what does this mean really since we don't have
  the code?}  \gb{The idea is that we need code to interpret
  adversaries. But the logic is sound for any interpretation of
  adversaries.}

External procedure calls, also known as adversary calls, are a bit
more involved. Each name $a \in \AProc$ is parametrized by a set
$\aora{a} \subseteq \IProc$ of internal procedures which the adversary may call, a
designated input variable $\farg{a} \in \VarL[a]$, a (unspecified)
piece of code $\fbody{a}$ that executes the function call, and a
result expression $\fret{a}$.  We assume that adversarial code can
only access its local variables in $\VarL[a]$ and can only make calls
to procedures in $\aora{a}$. It is possible to impose more
restrictions on adversaries---say, that they are lossless---but for
simplicity we do not impose additional assumptions on
adversaries here.

\section{Proof System} \label{sec:proofsystem}

In this section we introduce a program logic for proving properties of
probabilistic programs. The logic is abstract---assertions are arbitrary
predicates on sub-distributions---but the meta-theoretic properties are clearest
in this setting. In the following section, we will give a concrete version
suitable for practical use.

\paragraph*{Assertions and Closedness Conditions.}
We use predicates on state distribution.
\begin{definition}[Assertions]
  The set $\assn$ of assertions is defined as $\mathcal{P}(\Dist (\Mem))$. We
  write $\form(\state)$ for $\state \in \form$.
\end{definition}
Usual set operations are lifted to assertions using their logical
counterparts, e.g., $\form \land \form' \eqdef \form \cap \form'$
and $\neg \form \eqdef \overline{\form}$.
Our program logic uses a few additional constructions. Given a predicate
$\lasr$ over states, we define
\begin{gather*}
  \detm{\lasr}(\state) \eqdef \forall m.\, m
    \in \supp(\state) \implies \lasr(m)
\end{gather*}
where $\supp(\state)$ is the set of all states with non-zero probability under
$\state$.  Intuitively, $\lasr$ holds deterministically on all
states that we may sample from the distribution. To reason about branching
commands, given two assertions $\form_1$ and $\form_2$, we let
\begin{gather*}
 ( \form_1 \oplus \form_2)(\state) \eqdef
   \exists \state_1, \state_2 .\,
      \state = \state_1 + \state_2 \land \form_1(\state_1) \land \form_2(\state_2) 
      .
\end{gather*}
This assertion means that the sub-distribution is the sum
of two sub-distributions such that $\form_1$ holds on the first piece
and $\form_2$ holds on the second piece.

Given an assertion $\form$ and an event $E \subseteq \Mem$, we let
$
\drestr{\form}{E}(\state)\eqdef \form (\drestr{\state}{E}) .
$
This assertion holds exactly when $\form$ is true on the portion of the
sub-distribution satisfying $E$.  Finally, given an assertion $\form$ and a function
$F$ from $\Dist(\Mem)$ to $\Dist(\Mem)$, we define
$
  \form[F] \eqdef \lambda \state .\, \form(F(\state)) .
$
Intuitively, $\form[F]$ is true in a sub-distribution $\state$ exactly when
$\form$ holds on $F(\state)$.

Now, we can define the closedness properties of assertions. These
properties will be critical to our rules for $\kwhile$ loops.

\begin{definition}[Closedness properties]\label{def:closedness}
  A family of assertions $(\form_n)_{n\in\NNinf}$ is:
  \begin{itemize}
   \item
     $u$\emph{-closed} if for every increasing sequence of
     sub-distributions $(\state_n)_{n\in\NN}$ such that
     $\dsvalid{\form_n}{\state_n}$ for all $n\in\NN$ then
     $\dsvalid{\form_\infty}{\lim_{n\to\infty}\state_n}$;
    
  \item
    $t$\emph{-closed} if for every converging sequence of
    sub-distributions $(\state_n)_{n\in\NN}$ such that
    $\dsvalid{\form_n}{\state_n}$ for all $n\in\NN$ then
    $\dsvalid{\form_\infty}{\lim_{n\to\infty}\state_n}$;

  \item 
    $d$\emph{-closed} if it is $t$-closed and downward closed, that is
    for every sub-distributions $\state \leq \state'$,
    $\dsvalid{\form_\infty}{\state'}$ implies
    $\dsvalid{\form_\infty}{\state}$.
  \end{itemize}
  When $(\form_n)_n$ is constant  and equal to $\form$, we say that $\form$ is
  $u$-/$t$-/$d$-closed.
\end{definition}

Note that $t$-closedness implies $u$-closedness, but the converse does not
hold. Moreover, $u$-closed, $t$-closed and $d$-closed
assertions are closed under arbitrary intersections and finite unions,
or in logical terms under finite boolean combinations, universal
quantification over arbitrary sets and existential quantification over
finite sets.

%% There are simple examples where $t$-closedness fails for unbounded
%% expressions.\footnote{For example, consider the distribution of a
%%   state with a single variable and a sequence of distributions
%%   $(\zeta_i)_{i \in \NN^+}$ on $\NN$, where
%%       \[
%%         \zeta_i(j) = \left\{
%%           \begin{array}{ll}
%%             1/i &: j = i \\
%%             1 - 1/i &: j = 0 \\
%%             0 &: \text{otherwise} .
%%           \end{array}
%%         \right.
%%       \]
%%       Then, $\lim_{n \to \infty} \zeta_i = \zeta^*$, where $\zeta^*$
%%       returns $0$ with probability $1$.  However, $\pr{\zeta_i}{} = 1$ for
%%       all $i$, but $\pr{\zeta^*}{} \neq 1$, violating $t$-closedness.}

Finally, we introduce the necessary machinery for the frame rule. The
set $\MV(s)$ of \emph{modified} variables of a statement $s$ consists
of all the variables on the left of a deterministic or probabilistic
assignment.  In this setting, we say that an assertion $\form$ is
\emph{separated} from a set of variables $X$, written
$\aindep{\form}{X}$, if $\form(\state_1) \iff \form(\state_2)$ for any
distributions $\state_1$, $\state_2$ s.t.  $\wt{\state_1} =
\wt{\state_2}$ and ${\state_1}_{| \overline{X}} = {\state_2}_{|
  \overline{X}}$ where for a set of variables $X$, the restricted
  sub-distribution $\state_{| X}$ is
  \[
  \state_{| X} : m \in \Mem_{| X} \mapsto \Pr_{m' \sim \state} [m = m'_{| X}]
\]
where $\Mem_{| X}$ and $m_{| X}$ restrict $\Mem$ and $m$ to the variables in
$X$.

Intuitively, an assertion is separated from a set of variables
$X$ if every two sub-distributions that agree on the variables outside
$X$ either both satisfy the assertion, or both refute the assertion.

\paragraph*{Judgments and Proof Rules.}
Judgments are of the form $\hoare{\form}{s}{\form'}$, where the assertions
$\form$ and $\form'$ are drawn from  $\assn$.
\begin{definition}
A judgment $\hoare{\form}{s}{\form'}$ is \emph{valid}, written
  $\models \hoare{\form}{s}{\form'}$,
if $\dsvalid{\form'}{\dsem{\state}{s}}$ for every interpretation of
adversarial procedures and every probabilistic state $\state$ such
that $\dsvalid{\form}{\state}$.
\end{definition}

\Cref{fig:nonlooping:rules} describes the structural and basic rules
of the proof system.
Validity of judgments is preserved under standard structural rules,
like the rule of consequence \rname{Conseq}. As usual, the rule of
consequence allows to weaken the post-condition and to strengthen the
post-condition; in our system, this rule serves as the interface
between the program logic and mathematical theorems from probability
theory. The \rname{Exists} rule is helpful to deal with existentially
quantified pre-conditions.

The rules for $\skp$, assignments, random samplings and sequences are
all straightforward. The rule for $\abort$ requires $\detm{\bot}$ to
hold after execution; this assertion uniquely characterizes the
resulting null sub-distribution. The rules for assignments and random
samplings are semantical.

\begin{figure}
\begin{mathpar}
\inferrule[Conseq]
   {\form_0 \Rightarrow \form_1 \\
    \hoare{\form_1}{s}{\form_2} \\
    \form_2 \Rightarrow\form_3 }
   {\hoare{\form_0}{s}{\form_3}}
\and
\inferrule[Exists]
   {
    \forall x:T.\, \hoare{\form}{s}{\form'}}
   {\hoare{\exists x:T.\, \form}{s}{\form'}}
\and
\inferrule[Abort]
  { }
  {\hoare{\form}{\abort}{\detm{\bot}}} 
\and
\inferrule[Assgn]
  {\form' \eqdef \form[\dsem{}{x \asn e}]}
  {\hoare{\form'}{x \asn e}{\form}} 
\and
\inferrule[Skip]
  { }
  {\hoare{\form}{\skp}{\form}} 
\and
\inferrule[Sample]
  {\form' \eqdef \form[\dsem{}{x\rnd{g}}]}
  {\hoare{\form'}{x \rnd g}{\form}}
\and
\inferrule[Seq]
  {\hoare{\form_0}{s_1}{\form_1}\\
   \hoare{\form_1}{s_2}{\form_2}}
  {\hoare{\form_0}{s_1;s_2}{\form_2}}
\and
\inferrule[Cond]
  {\hoare{\form_1 \land \detm{e}}{s_1}{\form'_1}\\
   \hoare{\form_2 \land \detm{\neg e}}{s_2}{\form'_2}}
  {\hoare{(\form_1 \wedge \detm{e})\oplus (\form_2 \wedge \detm{\neg e})}
   {\ifstmt{e}{s_1}{s_2}}{\form'_1 \oplus \form'_2}}
\and
\inferrule[Split]
  {\hoare{\form_1}{s}{\form'_1} \\
   \hoare{\form_2}{s}{\form'_2}}
  {\hoare{\form_1 \oplus \form_2}{s}{\form'_1 \oplus \form'_2}}
\and
\inferrule[Frame]
  { \aindep{\form}{\MV(s)} \\ \mbox{$s$ is lossless} }
  { \hoare{\form}{s}{\form} }
\and
\inferrule[Call]
  {\hoare{\form}{\farg{f} \asn e; \fbody{f}}
         {\form'[\dsem{}{x \asn \fret{f}}]}}
  {\hoare{\form}{\call{x}{f}{e}}{\form'}}
\end{mathpar}

\caption{\label{fig:nonlooping:rules} Structural and basic rules}
\end{figure}

The rule \rname{Cond} for conditionals requires that the
post-condition must be of the form $\form_1\oplus\form_2$; this
reflects the semantics of conditionals, which splits the initial
probabilistic state depending on the guard, runs both branches, and
recombines the resulting two probabilistic states.

The next two rules (\rname{Split} and \rname{Frame}) are useful for
local reasoning. The \rname{Split} rule reflects the additivity of the
semantics and combines the \mbox{pre-} and post-conditions using the $\oplus$
operator.
The \rname{Frame} rule asserts that lossless statements preserve
assertions that are not influenced by modified variables.

The rule \rname{Call} for internal procedures is as expected, replacing the
procedure call $f$ with its definition.

\medskip

\Cref{fig:looping:rules} presents the rules for loops. We consider
four rules specialized to the termination behavior.
The \rname{While} rule is the most general rule, as it deals with
arbitrary loops. For simplicity, we explain the rule in the special
case where the family of assertions is constant, i.e.\ we have
$\form_n=\form$ and $\form'_n=\form'$. Informally, the
$\form$ is the loop invariant and $\form'$ is
an auxiliary assertion used to prove the invariant. We require that
$\form$ is $u$-closed, since the semantics of a loop is
defined as the limit of its lower approximations. Moreover, the first
premise ensures that starting from $\form$, one
guarded iteration of the loop establishes $\form'$; the
second premise ensures that restricting to $\neg e$ a
probabilistic state $\state'$ satisfying $\form'$ yields a
probabilistic state $\state$ satisfying $\form$. It is possible to
give an alternative formulation where the second premise
is substituted by the logical constraint $\drestr{\form'}{\neg
  e}\implies \form$. As usual, the post-condition of the loop is the
conjunction of the invariant with the negation of the guard (more
precisely in our setting, that the guard has probability 0).

The \rname{While-AST} rule deals with lossless loops. For
simplicity, we explain the rule in the special case where the family
of assertions is constant, i.e.\ we have $\form_n=\form$. In this
case, we know that lower approximations and approximations have the
same limit, so we can directly prove an invariant that holds after one
guarded iteration of the loop. On the other hand, we must now require
that the $\form$ satisfies the stronger property of
$t$-closedness.

The \rname{While-D} rule handles arbitrary loops with a
$d$-closed invariant; intuitively, restricting a sub-distribution
that satisfies a downwards closed assertion $\form$ yields a
sub-distribution which also satisfies $\form$. 

The \rname{While-CT} rule deals with certainly terminating loops. In
this case, there is no requirement on the assertions.

We briefly compare the rules from a verification perspective. If the
assertion is $d$-closed, then the rule \rname{While-D} is easier to
use, since there is no need to prove any termination requirement.
Alternatively, if we can prove certain termination of the loop, then
the rule \rname{While-CT} is the best to use since it does not impose
any condition on assertions.  When the loop is lossless, there
is no need to introduce an auxiliary assertion $\form'$, which
simplifies the proof goal. Note however that it might still be
beneficial to use the \rname{While} rule, even for lossless loops,
because of the weaker requirement that the invariant is $u$-closed
rather than $t$-closed.

\medskip

Finally, \cref{fig:adv:rules} gives the adversary rule for general adversaries.
It is highly similar to the general rule \rname{While-D} for loops since the
adversary may make an arbitrary sequence of calls to the oracles in $\aora{a}$
and may not be lossless. Intuitively, $\eta$ plays the role of the invariant: it
must be $d$-closed and it must be preserved by every oracle call with arbitrary
arguments. If this holds, then $\eta$ is also preserved by the adversary call.
Some framing conditions are required, similar to the ones of the
\textsc{[Frame]} rule: the invariant must not be influenced by the state
writable by the external procedures. 

It is possible to give other variants of the adversary rule with more
general invariants by restricting the
adversary, e.g., requiring losslessness or bounding the number of calls the
external procedure can make to oracles, leading to rules akin to the almost
surely terminating and certainly terminating loop rules, respectively.

\begin{figure*}
\begin{mathpar}
\inferrule[While]
  { 
    \uclosed{(\form'_n)_{n\in\NNinf}} \\\\
    \forall n.\, \hoare{\form_n}{\ift{e}{s}}{\form_{n+1}} \\
   \forall n.\, \hoare{\form_n}{\ift{e}{\abort}}{\form'_{n}}}
  {\hoare{\form_0}{\while{e}{s}}{\form'_\infty \wedge \detm{\neg e}}}
\\
\inferrule[While-AST]
          {\tclosed{(\form_n)_{n\in\NNinf}} \\
            \forall n.\, \hoare{\form_n}{\ift{e}{s}}{\form_{n+1}} \\ 
  \forall \state .\, \form_0(\state) \implies \wt{\dsem{\state}{(\while{e}{s})}}=1}
  {\hoare{\form_0}{\while{e}{s}}{\form_\infty \wedge \detm{\neg e}}}
\\
\inferrule[While-D]
          {\dclosed{(\form_n)_{n\in\NNinf}} \\
            \forall n.\, \hoare{\form_n}{\ift{e}{s}}{\form_{n+1}}}
  {\hoare{\form_0}{\while{e}{s}}{\form_\infty \wedge \detm{\neg e}}}
\\
\inferrule[While-CT]
          {\forall n.\, \hoare{\form_n}{\ift{e}{s}}{\form_{n+1}} \\
            \forall \state .\, \form_0(\state) \implies
   \dsem{\state}{(\ift{e}{s})^k} =  \dsem{\state}{(\while{e}{s})}}
  {\hoare{\form_0}{\while{e}{s}}{\form_k \wedge \detm{\neg e}}}
\end{mathpar}
\caption{\label{fig:looping:rules} Rules for loops}
\end{figure*}

\begin{figure*}
  \begin{mathpar}
\inferrule[Adv]
{\forall n\in\NNinf.~\aindep{\form_n}{\{x, \AVar\}} \\
    \dclosed{(\form_n)_{n\in\NNinf}} \\\\
   % \mbox{$\form$ is $d$-closed} \\\\
   % \mbox{$a$ is lossless} \\\\
   \forall f \in \aora{a}, x \in \VarL[a], e \in \expr, n \in \NN .\,
 \hoare {\form_n} {\call{x}{f}{e}} {\form_{n + 1}}}
 {\hoare {\form_0} {\call{x}{a}{e}} {\form_{\infty}}}
 \end{mathpar}
  \caption{\label{fig:adv:rules} Rules for adversaries}
\end{figure*}
  
\paragraph*{Soundness and Relative Completeness.}

Our proof system is sound and relatively complete with respect to the semantics;
these proofs have also been formalized in the \textsc{Coq} proof assistant.
\begin{theorem}[Soundness]
Every judgment $\hoare{\form}{s}{\form'}$ provable using the rules of
our logic is valid.
\end{theorem}
Completeness of the logic follows from the next lemma, whose proof
makes an essential use of the \rname{While} rule. In the sequel, we
use $\carac{\state}$ to denote the characteristic function of a
probabilistic state $\state$, an assertion stating that the current state is
equal to $\state$.
\begin{lemma}
For every probabilistic state $\state$, the following judgment is
provable using the rule of the logic:
\[ \hoare{\carac{\state}}{s}{\carac{\dsem{\state}{s}}}. \]
\end{lemma}

\begin{proof} By induction on the structure of $s$.
\begin{itemize}
\item $s = \abort$, $s = \skp$, $x \asn e$ and $s = x \rnd g$ are
  trivial;

\item $s = s_1;s_2$, we have to prove
  \[ \hoare{\carac{\state}}{s_1;s_2}
           {\carac{\dsem{\dsem{\state}{s_1}}{s_2}}} . \]
  We apply the \rname{Seq} rule with
  $\form_1 = \carac{\dsem{\state}{s_1}}$ premises can be
  directly proved using the induction hypothesis;

\item $s = \ifstmt{e}{s_1}{s_2}$, we have to prove 
  \[ \hoare{\carac{\state}}{\ifstmt{e}{s_1}{s_2}}
           {
               (\carac{\dsem{\drestr{\state}{e}}{s_1}} \oplus
                \carac{\dsem{\drestr{\state}{\neg e}}{s_2})}} . \]
  We apply the \rname{Conseq} rule to be able to apply the the
  \rname{Cond} rule with
  $\form_1 = \carac{\dsem{\drestr{\state}{e}}{s_1}}$ and
  $\form_2 = \carac{\dsem{\drestr{\state}{\neg e}}{s_2}}$ Both
  premises can be proved by an application of the \rname{Conseq} rule
  followed by the application of the induction hypothesis.
      
\item $s = \while{e}{s}$, we have to prove
  \[ \hoare{\carac{\state}}{\while{e}{s}}
           {\carac{
              \lim_{n \to \infty}\ \dsem{\state}{(\ift e s)^n;\ift e \abort}}} .\]
  We first apply the \rname{While} rule with
  $\form'_n = \carac{\dsem{\state}{(\ift e s)^n}}$ and
  \[
    \form_n = \carac{\dsem{\state}{(\ift e s)^n;\ift e \abort}} .
  \]
  For the first premise we apply the same process as for the
  conditional case: we apply the \rname{Conseq} and \rname{Cond}
  rules and we conclude using the induction hypothesis (and the
  \rname{Skip} rule).  For the second premise we follow the same
  process but we conclude using the \rname{Abort} rule instead of the
  induction hypothesis.  Finally we conclude since
  $\uclosed{(\form_n)_{n\in\NNinf}}$. \qedhere
\end{itemize}
\end{proof}

The abstract logic is also relatively complete. This property will be less
important for our purposes, but it serves as a basic sanity check.

\begin{theorem}[Relative completeness]
  Every valid judgment is derivable.
\end{theorem}

\begin{proof}
Consider a valid judgment $\hoare{\form}{s}{\form'}$.
Let $\state$ be a probabilistic state such that
$\dsvalid{\form}{\state}$. By the above proposition,
$ \hoare{\carac{\state}}{s}{\carac{\dsem{\state}{s}}}$. 
Using the validity of the judgment and \rname{Conseq}, we have
$\hoare{\carac{\state} \land \dsvalid{\form}{\state}}{s}{\form'}$.
Using the \rname{Exists} and \rname{Conseq} rules, we  conclude
$\hoare{\form}{s}{\form'}$
as required.
\end{proof}

The side-conditions in the loop rules (e.g.,
$\mathsf{uclosed}$/$\mathsf{tclosed}$/$\mathsf{dclosed}$ and the weight
conditions) are difficult to prove, since they are semantic properties. Next, we
present a concrete version of the logic with give easy-to-check, syntactic
sufficient conditions.

% ------------------------------------------------------------------------

\section{A Concrete Program Logic} \label{sec:concrete}

To give a more practical version of the logic, we begin by setting a concrete
syntax for assertions

\paragraph*{Assertions.} %\label{sec:assertions}
We use a two-level assertion language, presented in
\cref{fig:syntax}.  A \emph{probabilistic
assertion} $\pasr$ is a formula built from comparison of probabilistic
expressions, using first-order quantifiers and connectives, and the
special connective $\oplus$. A \emph{probabilistic expression} $p$ can be a
logical variable $v$, an operator applied to probabilistic expressions
$o(\vec{p})$ (constants are $0$-ary operators), or the
expectation $\ex{\lexp}$ of a state expression $\lexp$.
A \emph{state expression} $\lexp$ is either a program variable $x$, the
characteristic function $\ind{\lasr}$ of a state assertion $\lasr$, an
operator applied to state expressions $o(\vec{\lexp})$, or the
expectation $\dlet{v}{g}{\lexp}$ of state expression $\lexp$ in a
given distribution $g$. Finally, a \emph{state assertion} $\lasr$ is a
first-order formula over program variables.
Note that the set of operators is left unspecified but we assume that
all the expressions in $\expr$ and $\dexpr$ can be encoded by operators.

\iffull
\begin{figure}
\else
\begin{wrapfigure}{l}{0.55\textwidth}%[ht]
\fi
  \begin{align}
    % Logical expressions
      \lexp & ::= x
              \mid v
              \mid \ind{\lasr}
              \mid \dlet{v}{g}{\lexp}
              \mid o(\vec{\lexp})
              \tag{S-expr.} \\
    % Logical assertions
      \lasr & ::= \lexp \bowtie \lexp
              \mid FO(\lasr)
              \tag{S-assn.} \\
    % Probability expressions
      \pexp & ::= v \mid o(\vec{\pexp}) \mid \ex{\lexp} \tag{P-expr.} \\
    % Probability assertions
      \pasr & ::= \pexp \bowtie \pexp
              \mid \pasr \oplus \pasr 
              \mid FO(\pasr)
              \tag{P-assn.} \\
      \bowtie &\mathrel{\in} \{ \mathrel{=} , \mathrel{<}, \mathrel{\leq} \}
              \qquad o \in Ops
              \tag{Ops.}
  \end{align} 

\caption{\label{fig:syntax} Assertion syntax}
\iffull
\end{figure} 
\else
\end{wrapfigure} 
\fi

The interpretation of the concrete syntax is as expected. The interpretation of
probabilistic assertions is relative to a valuation $\rho$ which maps
logical variables to values, and is an element of $\assn$. The
definition of the interpretation is straightforward; the only
interesting case is $\pdenot{\ex{\lexp}}$ which is defined by
$\dlet{m}{\mu}{\mdenot{\lexp}}$, where $\mdenot{\lexp}$ is the
interpretation of the state expression $\lexp$ in the memory $m$ and
valuation $\rho$. The interpretation of state expressions is a mapping
from memories to values, which can be lifted to a mapping from
distributions over memories to distributions over values. The
definition of the interpretation is straightforward; the most
interesting case is for expectation $\mdenot{\dlet{v}{g}{\lexp}}
\eqdef \dlet{w}{\mdenot{g}}{\gdenot{\lexp}{\rho\subst{v}{w}}{m}}$.
We present the full interpretations in the supplemental materials.

Many standard concepts from probability theory have a natural
representation in our syntax. For example:
\begin{itemize}
  \item the probability that $\lasr$ holds in some probabilistic state
    is represented by the probabilistic expression $\Pr[\lasr]
    \eqdef \pr{\ind{\lasr}}{}$;
    
\item probabilistic independence of state expressions $\lexp_1$,
  \ldots, $\lexp_n$ is modeled by the probabilistic assertion $\indep
  \{\lexp_1, \ldots, \lexp_n\}$, defined by the clause\footnote{%
  The term $\Pr[\top]^{n - 1}$ is necessary since we work with sub-distributions.}
  $$\forall v_1 \ldots v_n,~ \Pr[\top]^{n - 1} \Pr[\bigwedge_{i=1\ldots
      n} \lexp_i = v_i] = \prod_{i=1\ldots n}\Pr[\lexp_i = v_i] ;
  $$

\item the fact that a distribution is proper is modeled by the probabilistic
  assertion $\llass \eqdef \Pr[\top] = 1$;
  
\item a state expression $\lexp$ distributed according to a law $g$ is
  modeled by the probabilistic assertion
  $$ \follows{\lexp}{g} \eqdef \forall w,~ \Pr[\lexp = w] = \ex{\dlet{v}{g}{\ind{v = w}}} .$$
  The inner expectation computes the probability that $v$ drawn from $g$ is
  equal to a fixed $w$; the outer expectation weights the inner probability by
  the probability of each value of $w$.
\end{itemize}
We can easily define $\square$ operator from the previous section in our new
syntax: $\detm{\lasr} \eqdef \Pr[\neg\lasr] = 0$.

\paragraph*{Syntactic Proof Rules.}

Now that we have a concrete syntax for assertions, we can give
syntactic versions of many of the existing proof rules. Such proof
rules are often easier to use since they avoid reasoning about the
semantics of commands and assertions. We tackle the non-looping rules
first, beginning with the following syntactic rules for assignment and
sampling:
\begin{mathpar}
  \inferrule[Assgn]
  { }
  {\hoare{\form\subst{x}{e}}{x \asn e}{\form}}
  \and
  \inferrule[Sample]
  { }
  {\hoare{\Samp{x}{g}(\form) }{x \rnd g}{\form}}
\end{mathpar}
The rule for assignment is the usual rule from Hoare logic,
replacing the program variable $x$ by its corresponding
expression $e$ in the pre-condition. The replacement $\form\subst{x}{e}$ is done recursively 
on the probabilistic assertion $\form$; for instance
for expectations, it is defined by 
$
\ex{\lexp}\subst{x}{e} \eqdef \ex{\lexp\subst{x}{e}} ,
$
where $\lexp\subst{x}{e}$ is the syntactic substitution.

The rule for sampling uses probabilistic substitution operator
$\Samp{x}{g}(\form)$, which replaces all occurrences of $x$ in $\form$ by a new
integration variable $t$ and records that $t$ is drawn from $g$; the operator is
defined in \cref{fig:samp}.

\iffull
\begin{figure}
\else
\begin{wrapfigure}{l}{0.45\textwidth}
\fi
  \[
  \begin{array}{lcll}
    \Samp{x}{g}(v) & \eqdef & v \\
    % \Samp{x}{g}(\dot{y}) & \eqdef & \dot{y} \\
    % \Samp{x}{g}(\pr{\int_{\overline{t}} \lasr\; d \overline{g}}) & \eqdef &
    % \pr{\int_{t_0}
    %   \int_{\overline{t}}\;
    %     (\lasr\; d \overline{g})[t_0 / x]\;
    %   d \mathcal{D}(e)} \\
    % \Samp{x}{g}(\pr{\int_{\overline{t}} \lexp\; d \overline{g}}) & \eqdef &
    % \pr{\int_{t_0}
    %   \int_{\overline{t}}\;
    %     (\lexp\; d \overline{g})[t_0 / x]\;
    %   d \mathcal{D}(e)} \\
    % \Samp{x}{g}(\int{\form}d\omega) & \eqdef &  \int{ \Samp{x}{g}(\phi) } d\omega \\
    \Samp{x}{g}\left(\pr{\lexp}{}\right) & \eqdef &
    \ex{\dlet{t}{g}{\lexp\subst{x}{t}}} \\
    \Samp{x}{g}(o(\vec{\form})) &\eqdef & 
                o(\Samp{x}{g}(\form_1), \ldots, \Samp{x}{g}(\form_n)) \\
    \Samp{x}{g}(\form_1 \bowtie \form_2) &\eqdef &
                \Samp{x}{g}(\form_1) \bowtie \Samp{x}{g}(\form_2)
    % \\
    % \Samp{x}{g}(\mathcal{Q}_{\dot{y}\in \basety}.\; \form) & \eqdef & 
    % \mathcal{Q}_{\dot{y}\in \basety}.\; \Samp{x}{g}(\form)
\end{array}
\]
for $o \in \Ops, \bowtie \in\{\wedge,\vee, \Rightarrow\}$.
\caption{Syntactic op. $\Samp{}{}$ (main cases) \label{fig:samp}}
\iffull
\end{figure}
\else
\end{wrapfigure}
\fi

\begin{figure*}
  \begin{align*}
    \sidecond{CTerm} & \eqdef \begin{array}[t]{l}
      \hoare{\llass\wedge\detm{(\tilde{e}=k\wedge 0< k \land b)}}{s}{
      \llass\wedge \detm{(\tilde{e} < k)}} \\
      \models \form \Rightarrow (\exists \dot{y}.\; \detm{\lexp \leq \dot{y}}) \land 
      \detm{(\tilde{e}=0 \Rightarrow \neg b)}
    \end{array}
    \\
    \sidecond{ASTerm} & \eqdef
    \begin{array}[t]{l}
      \hoare{\llass\wedge\detm{(\tilde{e}=k\wedge 0< k \leq K \land b)}}{s}{
        \llass\wedge \detm{(0\leq \tilde{e} \leq K )}\wedge 
      \Pr[\tilde{e} < k] \geq \epsilon} \\
      \models \form \Rightarrow \detm{(0 \leq \tilde{e} \leq K \land
      \tilde{e}=0 \Rightarrow \neg b)} \\
      \models \tclosed{\form} 
    \end{array}
  \end{align*}
  \caption{Side-conditions for loop rules} \label{fig:loop:syntactic}
\end{figure*}

Next, we turn to the loop rule. The side-conditions from
\cref{fig:looping:rules} are purely semantic, while in practice it is more
convenient to use a sufficient condition in the Hoare logic. We give
sufficient conditions for ensuring certain and almost-sure termination in
\cref{fig:loop:syntactic};
$\tilde{e}$ is an integer-valued
expression.  The first side-condition $\sidecond{CTerm}$ shows certain
termination given a strictly decreasing \emph{variant} $\tilde{e}$ that is
bounded below, similar to how a decreasing variant shows termination for
deterministic programs. The second side-condition $\sidecond{ASTerm}$ shows
almost-sure termination given a probabilistic variant $\tilde{e}$, which must be
bounded both above and below. While $\tilde{e}$ may increase with some
probability, it must decrease with strictly positive probability. This condition
was previously considered by \citet{HartSP83} for probabilistic transition
systems and also used in expectation-based
approaches~\citep{Morgan96:facs,Hurd03-jlap}. Our framework can also support
more refined conditions (e.g., based on super-martingales
\citep{ChakarovS13,mciver2016new}), but the condition $\sidecond{ASTerm}$
already suffices for most randomized algorithms.

While $t$-closedness is a semantic condition (cf. \cref{def:closedness}),
there are simple syntactic conditions to guarantee it. For instance, assertions
that carry a non-strict comparison $\bowtie \mathbin{\in} \{\leq,\geq,=\}$
between two bounded probabilistic expressions are $t$-closed;
the assertion stating probabilistic independence of a set of expressions is
$t$-closed.

\paragraph*{Precondition Calculus.}
With a concrete syntax for assertions, we are also able to incorporate syntactic
reasoning principles. One classic tool is Morgan and McIver's \emph{greatest
pre-expectation}, which we take as inspiration for a pre-condition calculus for
the loop-free fragment of \SYSTEM. Given an assertion $\pasr$ and a loop-free
statement $s$, we mechanically construct an assertion $\pasr^*$ that is the
pre-condition of $s$ that implies $\pasr$ as a post-condition. The basic idea is
to replace each expectation expression $p$ inside $\pasr$ by an expression $p^*$
that has the same denotation before running $s$ as $p$ after running $s$. This
process yields an assertion $\pasr^*$ that, interpreted before running $s$, is
logically equivalent to $\pasr$ interpreted after running $s$.

The computation rules for pre-conditions are defined in \cref{prem_wp}.  For a
probability assertion $\form$, its pre-condition $\wpre(s,\form)$ corresponds to
$\form$ where the expectation expressions of the form $\pr{\lexp}{}$ are
replaced by their corresponding \emph{pre-term}, $\prem(s,\pr{\lexp}{})$.
Pre-terms correspond loosely to Morgan and McIver's \emph{pre-expectations}---we
will make this correspondence more precise in the next section. The main
interesting cases for computing pre-terms are for random sampling and conditionals. For random sampling
the result is $\Samp{x}{g}(\pr{\lexp}{})$, which corresponds to the
[\textsc{Sample}] rule.  For conditionals, the expectation expression is split into
a part where $e$ is true and a part where $e$ is not true. We restrict the
expectation to a part satisfying $e$ with the operator
$
  \cond{\pr{\lexp}{}}{e}
  \eqdef
  \pr{\lexp \cdot \ind{e}}{} .
$
This corresponds to the expected value of $\lexp$ on the portion of the
distribution where $e$ is true.
Then, we can build the pre-condition calculus into \SYSTEM.
\begin{figure}
  \begin{align*}
    \prem(s_1; s_2, \pr{\lexp}{})
    &\eqdef
    \prem(s_1,\prem(s_2, \pr{\lexp}{}))
    \\
    \prem(x \asn e, \pr{\lexp}{})
    &\eqdef
    \pr{\lexp}{} \subst{x}{e}
    \\
    \prem(x \rnd g, \pr{\lexp}{})
    & \eqdef
    \Samp{x}{g}(\pr{\lexp}{})
    \\
    \prem(\ifstmt{e}{s_1}{s_2}, \pr{\lexp}{})
    & \eqdef 
    \cond{\prem({s_1}, \pr{\lexp}{})}{e} + 
    \cond{\prem({s_2}, \pr{\lexp}{})}{\neg e}
    \\[1em]
    \wpre(s, \pexp_1 \bowtie \pexp_2)
    &\eqdef
    \prem(s, \pexp_1) \bowtie \prem(s, \pexp_2)
  \end{align*}
  \caption{Precondition calculus (selected)}
  \label{prem_wp}
\end{figure}
\begin{thm} \label{thm:pc}
  Let $s$ be a non-looping command. Then, the following rule is derivable in the
  concrete version of \SYSTEM:
  \[
    \inferrule[PC]
    { }
    {\hoare{\wpre(s,\form)}{s}{\form}}
  \]
\end{thm}

% ------------------------------------------------------------------------

\section{Case Studies: Embedding Lightweight Logics}

While \SYSTEM is suitable for general-purpose reasoning about probabilistic
programs, in practice humans typically use more special-purpose proof
techniques---often targeting just a single, specific kind of property, like
probabilistic independence---when proving probabilistic assertions. When these
techniques apply, they can be a convenient and powerful tool.

To capture this intuitive style of reasoning, researchers have
considered lightweight program logics where the assertions and proof
rules are tailored to a specific proof technique. We demonstrate how
to integrate these tools in an assertion-based logic by introducing
and embedding a new logic for reasoning about independence and
distribution laws, useful properties when analyzing randomized
algorithms. We crucially rely on the rich assertions in \SYSTEM---it
is not clear how to extend expectation-based approaches to support
similar, lightweight reasoning. Then, we show to embed the union bound
logic \citep{BartheGGHS16-icalp} for proving accuracy bounds.

\subsection{Law and Independence Logic}
% Our second example is a proof system for reasoning about probabilistic
% independence and distribution laws. This type of reasoning is common
% when analyzing randomized algorithms; yet, it is particularly hard to
% capture formally. Many existing program logics for probabilistic programs
% either cannot capture these notions or provide poor support.
%
We begin by describing the law and independence logic \Sli,
a proof system with intuitive rules that are easy to
apply and amenable to automation. For simplicity, we only
consider programs which sample from the binomial distribution, and
have deterministic control flow---for lack of space, we also omit
procedure calls.
\begin{definition}[Assertions]
  \Sli assertions have the grammar:
  \[
    \begin{array}{rcl}
      \xi & := & \detE (e)
        \mid \indep E
        \mid \follows{e}{\binD(e,p)}
        \mid \top
        \mid \bot
        \mid \xi \land \xi
    \end{array}
  \]
  where $e\in\expr$, $E\subseteq\expr$, and $p\in [0,1]$.
\end{definition}
The assertion $\detE(e)$ states that $e$ is deterministic in the
current distribution, i.e., there is at most one element in the
support of its interpretation. The assertion $\indep E$ states that
the expressions in $E$ are independent, as formalized in the previous
section. The assertion $\follows{e}{\binD(m,p)}$ states that $e$ is
distributed according to a binomial distribution with parameter $m$
(where $m$ can be an expression) and constant probability $p$, i.e.\ the
probability that $e=k$ is equal to the probability that exactly $k$
independent coin flips return heads using a biased coin that returns heads with
probability $p$.

Assertions can be seen as an instance of a logical abstract domain,
where the order between assertions is given by implication based on a
small number of axioms. Examples of such axioms
include independence of singletons, irreflexivity of independence,
anti-monotonicity of independence, an axiom for the sum of
binomial distributions, and rules for deterministic expressions:
\begin{mathpar}
  \indep \{ x \}
  \and
  \indep \{ x, x\} \iff \detE(x)
  \and
  \indep (E\cup E') \implies \indep E
  \and
  %\left.\begin{array}{c}
  \follows{e\!}{\!\binD(m,p)} {\land} \follows{e'\!}{\!\binD(m',p)} {\land}
  \indep \{ e,e'\}\!\implies\!\mbox{$\follows{e{+}e'\!}{\!\binD(m+m',p)}$}
  \and
  \bigwedge_{1\leq i \leq n} \detE(e_i) \implies \detE(f(e_1,\ldots,e_n))
\end{mathpar}
\begin{definition}
  Judgments of the logic are of the form $\indhoare{\xi}{s}{\xi'}$,
  where $\xi$ and $\xi'$ are \Sli-assertions. A judgment is \emph{valid} if
  it is derivable from the rules of \cref{fig:fil}; structural rules
  and rule for sequential composition are similar to those from
  \Cref{sec:proofsystem} and omitted.
\end{definition}
The rule [\textsc{\Sli-Assgn}] for deterministic assignments is as in
\Cref{sec:proofsystem}. The rule [\textsc{\Sli-Sample}] for random
assignments yields as post-condition that the variable $x$ and a set
of expressions $E$ are independent assuming that $E$ is independent
before the sampling, and moreover that $x$ follows the law of the
distribution that it is sampled from. The rule
[\textsc{\Sli-Cond}] for conditionals requires that the guard is
deterministic, and that each of the branches satisfies the
specification; if the guard is not deterministic, there are
simple examples where the rule is not sound.\iffull\footnote{%
  Consider the following program where $\bernD(p)$ is the Bernoulli distribution
  with parameter $p$:
  \begin{align*}
    b \rnd \bernD(p); \kif\ {b}\ &  \kthen\ {x_1 \rnd \bernD(p_1); x_2 \rnd \bernD(p_2)} \\
    & \kelse\ {x_1 \rnd \bernD(p'_1); x_2 \rnd \bernD(p'_2)}
  \end{align*}
  Each branch establishes $\indep \{ x_1,x_2 \}$, but this is not a valid
  post-condition for the conditional. There are similar examples using the
binomial distribution.}\fi{}
The rule [\textsc{\Sli-While}] for loops requires that the loop is
certainly terminating with a deterministic guard.
Note that the requirement of certain termination could be avoided by
restricting the structural rules such that a statement $s$
has deterministic control flow whenever $\indhoare{\xi}{s}{\xi'}$
is derivable.

We now turn to the embedding. The embedding of \Sli assertions into
general assertions is immediate, except for $\detE(e)$ which is
translated as $\detm{e}\vee\detm{\neg e}$. We let $\overline{\xi}$
denote the translation of $\xi$.
\begin{thm}[Embedding and soundness of \Sli logic]
If $\indhoare{\xi}{s}{\xi'}$ is derivable in the \Sli
logic, then $\hoare{\overline{\xi}}{s}{\overline{\xi'}}$ is
derivable in (the syntactic variant of) \SYSTEM. As a consequence,
every derivable judgment $\indhoare{\xi}{s}{\xi'}$ is valid.
\end{thm}
\begin{proof}[Proof sketch] By induction on the derivation.
  The interesting cases are conditionals and loops. For
  conditionals, the soundness follows from the soundness of the rule:
  $$
  \inferrule{\hoare{\form}{s_1}{\form'}\\
    \hoare{\form}{s_2}{\form'} \\
    \detm{e}\vee\detm{\neg e}}
  {\hoare{\form}{\ifstmt{e}{s_1}{s_2}}{\form'}}
  $$
To prove the soundness of this rule, we proceed by case analysis on
$\detm{e} \vee\detm{\neg e}$. We treat the case $\detm{e}$; the other
case is similar. In this case, $\form$ is equivalent to $\form_1\wedge
\detm{e} \oplus\form_2\wedge\detm{\neg e}$, where $\form_1=\form$ and
$\form_2=\bot$. Let $\form_1'=\form'$ and $\form_2=\detm{\bot}$;
again, $\form_1'\oplus\form_2'$ is logically equivalent to
$\form'$. The soundness of the rule thus follows from the soundness of
the [\textsc{Cond}] and [\textsc{Conseq}] rules.
For loops, there exists a natural number $n$ such that $\while{b}{s}$ is
semantically equivalent to $(\ift{b}{s})^n$.  By assumption
$\indhoare{\xi}{s}{\xi}$ holds, and thus by induction hypothesis
$\hoare{\overline{\xi}}{s}{\overline{\xi}}$. We also have
$\xi\implies\detE(b)$, and hence
$\hoare{\overline{\xi}}{\ift{b}{s}}{\overline{\xi}}$. We conclude by [\textsc{Seq}].
\end{proof}

\begin{figure}
\begin{mathpar}
\inferrule[\Sli-Assgn]{~}
     {\indhoare{\xi\subst{x}{e}}{x \asn e}{\xi}} 
\and
\inferrule[\Sli-Sample]
  {\{x \} \cap \Vars(E) \cap \Vars(e) = \emptyset}
  {\indhoare{\indep E}{x \rnd \binD(e,p)}{\indep (E \cup \{x\})
      \land \follows{x}{\binD(e,p)}}}
\and
\inferrule[\Sli-Seq]
    {\indhoare{\xi}{s_1}{\xi'}\\
     \indhoare{\xi'}{s_2}{\xi''}}
    {\indhoare {\xi}{s_1;s_2}{\xi''}}
\and
\inferrule[\Sli-Cond]
    {\indhoare{\xi}{s_1}{\xi'}\\
     \indhoare{\xi}{s_2}{\xi'}\\\\
     \xi \implies\detE(b) }
    {\indhoare {\xi}{\ifte{b}{s_1}{s_2}}{\xi'}}
\and
\inferrule[\Sli-While]
  {\indhoare{\xi}{s}{\xi} \\ 
   \xi\implies\detE(b) \\
   \sidecond{CTerm} }
  {\indhoare {\xi}{\while{b}{s}}{\xi}}
\end{mathpar}

\caption{\Sli proof rules (selected)}
\label{fig:fil} 
\end{figure}

To illustrate our system \Sli, consider the statement $s$ in \cref{fig:binsum} which
flips a fair coin $N$ times and counts the number of heads.
Using the logic, we prove that $\follows{\lstt{c}}{\binD(N \cdot (N +
1)/2, {1}/{2})}$ is a post-condition for $s$. We take the invariant:
$$\follows{\lstt{c}}{\binD \left({\lstt{j}(\lstt{j}+1)}/{2},{1}/{2}\right)}$$
The invariant holds initially, as $\follows{0}{B(0,{1}/{2})}$.
For the inductive case, we show:
$$\indhoare{\follows{\lstt{c}}{\binD\left(0,{1}/{2}\right)}}{s_0}{
  \follows{\lstt{c}}{\binD \left({(\lstt{j}+1)(\lstt{j}+2)}/{2},{1}/{2}\right)}}$$
where $s_0$ represents the loop body, i.e.\ $\lstt{x}\rnd
\binD\left(\lstt{j},{1}/{2}\right);\lstt{c} \asn \lstt{c} + \lstt{x}$. First, we apply the
rule for sequence taking as intermediate assertion
$$\follows{\lstt{c}}{\binD
  \left({\lstt{j}(\lstt{j}+1)}/{2},{1}/{2}\right)} \land
\follows{\lstt{x}}{\binD\left(\lstt{j},{1}/{2}\right)}
\land \indep \{\lstt{x},\lstt{c}\}$$
\iffull
\begin{figure}
\else
\begin{wrapfigure}{l}{0.3\textwidth}
\fi
\begin{lstlisting}
proc sum () = 
  var c:int, x:int; 
  c := 0; 
  for j := 1 to N do
    x ~~ B(j,1/2);
    c := c + x;
  return c
\end{lstlisting}
\caption{Sum of bin.}\label{fig:binsum}
\iffull
\end{figure}
\else
\end{wrapfigure}
\fi

The first premise follows from the rule for random assignment and structural
rules. The second premise follows from the rule for deterministic assignment and
the rule of consequence, applying axioms about sums of binomial distributions.

We briefly comment on several limitations of \Sli.  First, \Sli is
restricted to programs with deterministic control flow, but this
restriction could be partially relaxed by enriching \Sli with
assertions for conditional independence. Such assertions are already
expressible in the logic of \SYSTEM; adding conditional independence
would significantly broaden the scope of the \Sli proof system and
open the possibility to rely on axiomatizations of conditional
independence (e.g., based on graphoids~\citep{PearlP86}). Second, the
logic only supports sampling from binomial distributions. It is
possible to enrich the language of assertions with clauses
$\follows{c}{g}$ where $g$ can model other distributions, like the
uniform distribution or the Laplace distribution. The main design
challenge is finding a core set of useful facts about these
distributions. Enriching the logic and automating the analysis are
interesting avenues for further work.

\subsection{Embedding the Union Bound Logic}
The program logic \Sahl \citep{BartheGGHS16-icalp} was recently introduced for
estimating accuracy of randomized computations. One main application of \Sahl
is proving accuracy of randomized algorithms, both in the offline and online
settings---i.e.\ with adversary calls.
\Sahl is based on the union bound, a basic tool from probability
theory, and has judgments of the form
$
\ubhl \beta \Phi s \Psi ,
$
where $s$ is a statement, $\Phi$ and $\Psi$ are first-order formulae
over program variables, and $\beta$ is a probability, i.e.\,
$\beta\in{[0,1]}$. A judgment $\ubhl \beta \Phi s \Psi$ is valid if
for every memory $m$ such that $\Phi(m)$, the probability of $\neg
\Psi$ in $\dsem{m}{s}$ is upper bounded by $\beta$, i.e.\,
$\Pr\!_{\dsem{m}{s}}[\neg \Psi]\leq \beta$.

\Cref{fig:ubhl} presents some key rules of \Sahl, including a rule for
sampling from the Laplace distribution $\mathcal{L}_\epsilon$ centered
around $e$. The predicate $\sidecond{CTerm}(k)$ indicates that the
loop terminates in at most $k$ steps on any memory that satisfies the
pre-condition. Moreover, $\beta$ is a function of $\epsilon$.
\begin{figure}
\begin{mathpar}
  \inferrule[\Sahl-Sample]
  { }
  { \ubhl{\beta}
    {\top}
    {{x}\rnd {\mathcal{L}_\epsilon(e)}}
    { |x - e| \leq \frac{1}{\epsilon} \log \frac{1}{\beta} } }
\\
\inferrule[\Sahl-Seq]
  {\ubhl{\beta_1}{\Phi}{s_1}{\Theta}  \qquad \ubhl{\beta_2}{\Theta}{s_2}{\Psi} }
  {\ubhl{\beta_1+\beta_2}{\Phi}{s_1;s_2}{\Psi} }
\\
\inferrule[\Sahl-While]
  {\ubhl{\beta}{\Phi}{c}{\Phi} \\
   \sidecond{CTerm}(k)}
  { \ubhl{k \cdot \beta}{\Phi}{ \while{e}{c} }{\Phi \land \neg e} }
\end{mathpar}
\caption{\Sahl proof rules (selected)}
\label{fig:ubhl}
\end{figure}

\Sahl has a simple embedding into \SYSTEM.

\begin{thm}[Embedding of \Sahl]
If $\ubhl{\beta}{\Phi}{s}{\Psi}$ is derivable in \Sahl, then $\hoare
{\detm{\Phi}}{s} {\pr{\ind{\neg \Psi}}{} \leq \beta}$ is derivable in \SYSTEM.
\end{thm}
% ------------------------------------------------------------------------

\section{Case Studies: Verifying Randomized Algorithms}\label{sec:examples}

In this section, we will demonstrate \SYSTEM on a selection of
examples; we present further examples in the supplemental material.
Together, they exhibit a wide variety of different proof techniques
and reasoning principles which are available in the \SYSTEM's
implementation.

\paragraph*{Hypercube Routing.}
 will begin with the \emph{hypercube routing} algorithm
\citep{valiant1982scheme,Valiant:1981:USP:800076.802479}.  Consider a
network topology (the \emph{hypercube}) where each node is labeled by
a bitstring of length $D$ and two nodes are connected by an edge if
and only if the two corresponding labels differ in exactly one bit
position.

% This network topology is known as a
% \emph{hypercube}, a $D$-dimensional version of the standard cube.
% ; a simple
% example with $D = 3$ is in \cref{fig:hypercube}.  
% \begin{wrapfigure}{l}{0.2\textwidth}
% %\begin{center}
%   \includegraphics[width=0.2\textwidth]{hypercube}
%   \caption{Hypercube path from $111$ to $100$ ($D = 3$)}
%   \label{fig:hypercube}
% %\end{center}
% \end{wrapfigure}

%  moving the packets from node to node, following the edges,
% until all packets reach their destination. Time proceeds in a series
% of steps, and at each step at most one packet can traverse any single
% edge. If several packets try to use the same edge simultaneously, one
% packet will be selected to move (for any selection strategy that
% selects \emph{some} packet). The other packets wait at their current
% position; these packets are \emph{delayed} and make no progress this
% step.

In the network, there is initially one packet at each node, and each
packet has a unique destination. The algorithm implements a routing
strategy based on \emph{bit fixing}: if the current position has
bitstring $i$, and the target node has bitstring $j$, we compare the
bits in $i$ and $j$ from left to right, moving along the edge that
corrects the first differing bit.  Valiant's algorithm uses
randomization to guarantee that the total number of steps
grows \emph{logarithmically} in the number of packets.  In the first
phase, each packet $i$ select an intermediate destination
$\rho(i)$ uniformly at random, and use bit fixing to
reach $\rho(i)$. In the second phase, each packet use bit fixing
to go from $\rho(i)$ to the destination $j$.  We will focus on
the first phase since the reasoning for the second phase is nearly
identical.  We can model the strategy with the code in
\Cref{fig:hypercube}, using some syntactic sugar for the $\kfor$
loops.\footnote{%
  Recall that the number of node in a hypercube of dimension $D$ is $2^D$ so
  each node can be identified by a number in $[1,2^D]$.}
\iffull
\begin{figure}
\else
\begin{wrapfigure}{l}{0.47\textwidth}
\fi
\begin{lstlisting}
proc route ($D$ $T$ : int) :
  var $\rho$, pos, usedBy : node map;
  var nextE : edge;
  pos  :=  Map.init id $2^D$; $\rho$  :=   Map.empty;
  for $i$  :=  1 to $2^D$ do
    $\rho$[i] ~~ $[1,2^D]$
    for $t$  :=  1 to $T$ do
      usedBy  :=  Map.empty;
      for $i$  :=  1 to $2^D$ do
        if pos$[i] \neq \rho[i]$ then
          nextE  :=  getEdge pos[$i$] $\rho[i]$;
          if usedBy[nextE] = $\bot$ then
            // Mark edge used
            usedBy[nextE]  :=  $i$;
            // Move packet
            pos[$i$]  :=  dest nextE
  return (pos, $\rho$)
\end{lstlisting}
\caption{Hypercube Routing}\label{fig:hypercube}
\iffull
\end{figure}
\else
\end{wrapfigure}
\fi
We assume that initially, the position of the packet $i$ is at node $i$
(see \lstt{Map.init}). Then, we initialize the random intermediate
destinations $\rho$. The remaining loop encodes the evaluation of the
routing strategy iterated $T$ time.  The variable \lstt{usedBy} is a
map that logs if an edge is already used by a packet, it is empty at
the beginning of each iteration.  For each packet, we try to move it
across one edge along the path to its intermediate destination.  The
function \lstt{getEdge} returns the next edge to follow, following the
bit-fixing scheme.  If the packet can progress (its edge is not used),
then its current position is updated and the edge is marked as used.

We show that if the number of timesteps $T$ is $4 D + 1$, then all
packets reach their intermediate destination in at most $T$ steps,
except with a small probability $2^{-2D}$ of failure. That is, the
number of timesteps grows linearly in $D$, logarithmic in the number
of packets. This is formalized in our system as:
\[
  \hoare{T = 4D + 1}{\lstt{route}}{\Pr[ \exists i.\; \lstt{pos}[i] \neq
    \lstt{$\rho$}[i] ] \leq 2^{-2D}]}
\]

\paragraph*{Modeling Infinite Processes.}
\iffull
\begin{figure}
\else
\begin{wrapfigure}{l}{0.4\textwidth}
\fi
\begin{lstlisting}
proc coupon ($N$ : int) :
  var int cp[$N$], $t$[$N$];
  var int $X$ :=   0;
  for $p$  :=  1 to $N$ do
    ct  :=  0;
    cur ~~ $[1,N]$;
    while cp[cur] = 1 do
      ct  :=  ct + 1;
      cur ~~ $[1,N]$;
    $t$[$p$]       :=  ct;
    cp[cur]  :=  1;
    $X$  :=  $X$ + $t$[$p$];
  return $X$
\end{lstlisting}
\caption{Coupon collector}\label{fig:coupon}
\iffull
\end{figure}
\else
\end{wrapfigure}
\fi
Our second example is the \emph{coupon collector} process. The algorithm draws a
uniformly random coupon (we have $N$ coupon) on each day, terminating when it
has drawn at least one of each kind of coupon.  The code of the algorithm is
displayed in \cref{fig:coupon}; the array \lstt{cp} records of the coupons seen
so far, $\lstt{t}$ holds the number of steps taken before seeing a new coupon,
and $X$ tracks of the total number of steps.  Our goal is to bound the average
number of iterations. This is formalized in our logic as:
\[
  \{ \llass \}\;
  \lstt{coupon} \;
  \left\{
      \ex{X} = \textstyle\sum_{i \in [1,N]} \left(
        \frac{N}{N - i + 1} \right)
  \right\} .
\]

%% \jh{TODO: This comes out of nowhere. Either need to explain why we compare with
%%   [24] or drop the comparison.}
%% \gb{decided to drop the comparison since it talks again about independence}
%% \paragraph*{Comparison with \citet{KaminskiKMO16}.}
%% % Coupon collector provides a point of comparison with
%% % expectation-based techniques, as \citet{KaminskiKMO16} have
%% % independently proved the example using a system inspired by \Spgcl.
%% At a high level, their analysis involves complex equations instead of our
%% invariants. The two approaches are difficult to compare, in the sense that it
%% seems hard to infer their equations from our invariants and vice-versa. However,
%% we believe that the logical invariants used by our proof are more intuitive and
%% easier to discover than the equations used in the \Spgcl proof.  Moreover, we
%% note that \SYSTEM naturally supports reasoning about independence, and could be
%% used to reason about the variance of $X$, whereas reasoning about independence
%% in \Spgcl remains possible but challenging.

\paragraph*{Limited Randomness.}
\iffull
\begin{figure}
\else
\begin{wrapfigure}{r}{.45\textwidth}
\fi
\begin{lstlisting}
proc pwInd (N : int) :
  var bool X[2${}^{\mbox N}$], B[N];
  for i :=  1 to N do
    B[i] ~~ Ber(1/2);
  for j  :=  1 to 2${}^{\mbox N}$ do
    X[j]  :=  0;
    for k  :=  1 to N do
      if k $\in$ bits(j) then
        X[j]  :=  X[j] $\oplus$ B[k]
  return X
\end{lstlisting}
\caption{Pairwise Independence}\label{fig:pwindep}
\iffull
\end{figure}
\else
\end{wrapfigure}
\fi
\emph{Pairwise independence} says that if we see the result of
$X_i$, we do not gain information about all other variables
$X_k$. However, if we see the result of \emph{two} variables $X_i,
X_j$, we may gain information about $X_k$.
There are many constructions in the algorithms literature that grow
a small number of independent bits into more pairwise
independent bits. \Cref{fig:pwindep} gives one procedure, where
$\oplus$ is exclusive-or, and $\lstt{bits(j)}$ is the set of 
positions set to $1$ in the binary expansion of $\lstt{j}$.
The proof uses the following fact, which we fully verify:
for a uniformly distributed Boolean random variable $Y$, and a random
variable $Z$ of any type,
\begin{equation} \label{eq:pw-indep}
  Y \mathrel{\indep} Z \Rightarrow Y \oplus f(Z) \mathrel{\indep} g(Z)
\end{equation}
for any two Boolean functions $f, g$. Then, note that
$\lstt{X}[i] = \bigoplus_{\{j \in \lstts{bits}(i)\}} \lstt{B}[j]$
where the big XOR operator ranges over the indices $j$ where the bit
representation of $i$ has bit $j$ set. For any two $i, k \in [1, \dots,
2^{\lstts{N}}]$ distinct, there is a bit position in $[1, \dots, \lstt{N}]$
where $i$ and $k$ differ; call this position $r$ and suppose it is set in $i$
but not in $k$. By rewriting,
\[
  \lstt{X}[i] = \lstt{B}[r] \oplus \bigoplus_{\{j \in \lstts{bits}(i) \setminus
    r\}} \lstt{B}[j]
  \quad \text{and} \quad
  \lstt{X}[k] = \bigoplus_{\{j \in \lstts{bits}(k) \setminus r\}} \lstt{B}[j] .
\]
Since $\lstt{B}[j]$ are all independent, $\lstt{X}[i] \mathrel{\indep} \lstt{X}[k]$ follows from
\cref{eq:pw-indep} taking $Z$ to be the distribution on tuples $\langle
\lstt{B}[1],
\dots, \lstt{B}[\lstt{N}] \rangle$ excluding $\lstt{B}[r]$. This verifies
pairwise independence:
$$
\hoare{\llass}
{\lstt{pwInd(N)}}
{\llass \land \forall i , k \in [2^{\lstts{N}}] . \; i \neq k \Rightarrow
  \lstt{X}[i] \mathrel{\indep} \lstt{X}[k]  } .
$$

\paragraph*{Adversarial Programs.}
Pseudorandom functions (PRF) and pseudorandom permutations (PRP) are two
idealized primitives that play a central role in the design of symmetric-key
systems. Although the most natural assumption to make about a blockcipher is
that it behaves as a pseudorandom permutation, most commonly the security of
such a system is analyzed by replacing the blockcipher with a perfectly random
function. The PRP/PRF Switching Lemma
\cite{DBLP:conf/stoc/ImpagliazzoR89,DBLP:conf/eurocrypt/BellareR06} fills the
gap: given a bound for the security of a blockcipher as a pseudorandom function,
it gives a bound for its security as a pseudorandom permutation.

\begin{lem}[PRP/PRF switching lemma]
Let $A$ be an adversary with blackbox access to an oracle O implementing 
either a random permutation on $\bs{l}$ or a random function from $\bs{l}$
to $\bs{l}$. Then the probability that the adversary $A$ 
distinguishes between the two oracles in at most $q$ calls is
bounded by
$$|\Pr_{\textrm{PRP}}[b \land |H| \leq q] - 
   \Pr_{\textrm{PRF}}[b \land |H| \leq q] | \leq \frac{q(q-1)}{2^{l+1}} ,$$
where $H$ is a map storing each adversary call and $|H|$ is its size.
\end{lem}

Proving this lemma can be done using the Fundamental Lemma of
Game-Playing, and bounding the probability of \emph{bad} in the
program from \cref{fig:prp/prf}. We focus on the latter. Here we apply
the [\textsc{Adv}] rule of \SYSTEM with the invariant
$\forall k, \Pr[\lstt{bad} \land |H| \leq k] \leq
\frac{k(k-1)}{2^{l+1}}$
where $|H|$ is the size of the map $H$, i.e. the number 
of adversary call.
Intuitively, the invariant says that at each call to the oracle 
the probability that $\lstt{bad}$ has been set before and
that the number of adversary call is less than $k$ is bounded 
by a polynomial in $k$. 

The invariant is $d$-closed and true before the
adversary call, since at that point $\Pr[\lstt{bad}] = 0$.  Then we
need to prove that the oracle preserves the invariant, which can be
done easily using the precondition calculus ([\textsc{PC}] rule).

\begin{figure}
  \begin{minipage}{0.6\linewidth}
\begin{lstlisting}
var $H$: ($\bs{l}$, $\bs{l}$) map;
proc orcl ($q$:$\bs{l}$):
  var $a:\bs{l}$;
  if $q \not\in H$ then
    $a$    ~~ $\bs{l}$;
    bad  :=  bad || $a \in \textrm{codom}(H)$;
    $H[q]$  :=  $a$;
  return $H[q]$;
\end{lstlisting}
\end{minipage}
\begin{minipage}{0.39\linewidth}
\begin{lstlisting}
proc main():
  var $b$: bool;
  bad  :=  false;
  $H$    :=  [];
  $b$    :=  A();
  return $b$;
\end{lstlisting}
\end{minipage}
\caption{PRP/PRF game}\label{fig:prp/prf}
\end{figure}

% ------------------------------------------------------------------------
\section{Implementation and Mechanization}

We have built a prototype implementation of \SYSTEM within
\EasyCrypt~\citep{BartheGHZ11,BartheDGKSS13}, a theorem prover
originally designed for verifying cryptographic protocols. \EasyCrypt
provides a convenient environment for constructing proofs in various
Hoare logics, supporting interactive, tactic-based proofs for
manipulating assertions and allowing users to invoke external tools,
like SMT-solvers, to discharge proof obligations. \EasyCrypt provides
a mature set of libraries for both data structures (sets, maps, lists,
arrays, etc.) and mathematical theorems (algebra, real analysis,
etc.), which we extended with theorems from probability theory.

\begin{wraptable}{l}{0.4\textwidth}
\begin{tabular}{lcc}
  \toprule
  Example & LC & FPLC \\
  \midrule
  {\tt hypercube       } & 100 & 1140 \\
  {\tt coupon          } & 27  & 184  \\
  {\tt vertex-cover    } & 30  & 61   \\
  {\tt pairwise-indep  } & 30  & 231  \\
  {\tt private-sums    } & 22  & 80   \\
  {\tt poly-id-test    } & 22  & 32   \\
  {\tt random-walk     } & 16  & 42   \\
  {\tt dice-sampling   } & 10  & 64   \\
  {\tt matrix-prod-test} & 20  & 75   \\
  \bottomrule
\end{tabular}
\caption{Benchmarks \label{tab:examples}}
\end{wraptable}

% \paragraph*{Formal Verification of Examples.}
We used the implementation for verifying many examples from the
literature, including all the programs presented in
\cref{sec:examples} as well as some additional examples (such as
polynomial identity test, private running sums, properties about
random walks, etc.). The verified proofs bear a strong resemblance to
the existing, paper proofs.  Independently of this work, \SYSTEM has
been used to formalize the main theorem about a randomized
gossip-based protocol for distributed systems \citep[Theorem
2.1]{KempeDG03}. Some libraries developed in the scope of \SYSTEM
have been incorporated into the main branch of \EasyCrypt, including
a general library on probabilistic independence.

% \paragraph*{First-Order Treatment of Memories.}
% \EasyCrypt refers to program states as memories, and uses a special
% type \textbf{Mem} to denote memories. While the notion of memory is
% central to \EasyCrypt, the type \textbf{Mem} does not have a
% first-class status: any definition or construct that involves the type
% \textbf{Mem} must be added by extending the \EasyCrypt kernel. In
% contrast, \SYSTEM allows the definition of operators and predicates
% that depend on memories, making the type of memories no more special
% than the type of booleans. This transversal change enables
% formalization of probabilistic states, assertions, \emph{etc}, from
% first principles within the ambient logic of \SYSTEM.

\paragraph*{A New Library for Probabilistic Independence.}
In order to support assertions of the concrete program logic, we enhanced
the standard libraries of \EasyCrypt, notably the ones dealing with
big operators and sub-distributions. Like all \EasyCrypt libraries, they are
written in a foundational style, i.e.\, they are defined instead of axiomatized.
A large part
of our libraries are proved formally from first principles. However,
some results, such as concentration bounds, are currently declared as
axioms.

\jh{Check: are concentration bounds axiomatized?}

Our formalization of probabilistic independence deserves special
mention. We formalized two different (but logically equivalent)
notions of independence. The first is in terms of products of
probabilities, and is based on heterogenous lists. Since \SYSTEM (like
\EasyCrypt) has no support for heterogeneous lists, we use a smart
encoding based on second-order predicates. The second definition is
more abstract, in terms of product and marginal distributions. While
the first definition is easier to use when reasoning about randomized
algorithms, the second definition is more suited for proving
mathematical facts. We prove the two definitions equivalent, and
formalize a collection of related theorems.

% \paragraph*{Integrating \EasyCrypt and \SYSTEM.}
% Although \EasyCrypt and \SYSTEM are closely related, they are not
% fully integrated at this time. There are obvious benefits to their
% integration: many cryptographic proofs have conditional equivalence
% (\lq\lq up to bad\rq\rq) steps, where we want to bound the probability
% of events---this can be done with \SYSTEM. Conversely, many proofs of
% randomized algorithms do not follow the original program code closely,
% and would be easier to prove on a functionally equivalent program that
% more closely follows the line of reasoning---\EasyCrypt was designed
% to prove such equivalences. More generally, many examples may benefit
% from combined relational and non-relational reasoning.

\paragraph*{Mechanized Meta-Theory.}
The proofs of soundness and relative completeness of the
abstract logic, without adversary calls, and the syntactical
termination arguments have been mechanized in the \textsf{Coq} proof
assistant. The development is available in supplemental material.

\section{Related Work}\label{sec:related}

\paragraph*{More on Assertion-Based Techniques.}
The earliest assertion-based system is due to
Ramshaw~\citep{Ramshaw79}, who proposes a program logic where
assertions can be formulas involving \emph{frequencies}, essentially
probabilities on sub-distributions. Ramshaw's logic allows assertions
to be combined with operators like $\oplus$, similar to our
approach. \citet{Hartog:thesis} presents a Hoare-style logic with
general assertions on the distribution, allowing expected values and
probabilities. However, his \textbf{while} rule is based on a semantic
condition on the guarded loop body, which is less desirable for
verification because it requires reasoning about the semantics of
programs. \citet{ChadhaCMS07} give decidability results for a
probabilistic Hoare logic without \textbf{while} loops. We are not
aware of any existing system that supports assertions about general
expected values; existing works also restrict to Boolean
distributions. \citet{RandZ15} formalize a Hoare logic for
probabilistic programs but unlike our work, their assertions are
interpreted on \emph{distributions} rather than sub-distributions.
For conditionals, their semantics rescales the distribution of states
that enter each branch. However, their assertion language is limited
and they impose strong restrictions on loops.

% \paragraph*{Formalizations of Probability Theory.}
% Formalizations of measure and integration theory in general purpose
% interactive theorem provers have been considered in many works
% \citep{AudebaudP09,Hurd03,Richter04,Coble10,MhamdiHT10,HolzlH11}.
% \citet{AvigadHS14} recently completed a proof of the Central Limit
% theorem, which is the principle underlying concentration bounds.
% \citet{holzl2016markov} formalized discrete-time Markov chains and
% Markov Decision Processes. These, and other, existing formalizations
% have been used to verify several case studies, but they are scattered
% and not easily accessible for our purposes.

\paragraph*{Other Approaches.}
Researchers have proposed many other approaches to verify probabilistic
program. For instance, verification
of Markov transition systems goes back to at least
\citet{HartSP83,SharirPH84}; our condition for ensuring almost-sure
termination in loops is directly inspired by their work. Automated methods
include model checking (see e.g., \citep{Baier16,Katoen16,KNP11}) and abstract
interpretation (see e.g., \citep{Monniaux00,CousotM12}). Techniques for
reasoning about higher-order (functional) probabilistic languages are an active
subject of research (see e.g.,
\citep{bizjak2015step,10.1007/978-3-642-54833-8_12,DalLago:2014:CEH:2535838.2535872}).
For analyzing probabilistic loops, in particular, there are tools for
reasoning about running time.  There are also automated systems for
synthesizing invariants \citep{ChatterjeeFNH16,BEFFH16}.
\citet{ChakarovS13,ChakarovS14} use a martingale method to compute the
expected time of the coupon collector process for $N=5$---fixing $N$
lets them focus on a program where the outer \textbf{while} loop is
fully unrolled. Martingales are also used by \citet{FioritiH15} for
analyzing probabilistic termination.  Finally, there are approaches
involving symbolic execution; \citet{SampsonPMMGC14} use a mix of
static and dynamic analysis to check probabilistic programs from the
approximate computing literature.

\section{Conclusion and Perspectives}
We introduced an expressive program logic for probabilistic programs, and showed
that assertion-based systems are suited for practical verification of
probabilistic programs. Owing to their richer assertions, program logics are a
more suitable foundation for specialized reasoning principles than
expectation-based systems. As evidence, our program logic can be smoothly
extended with custom reasoning for probabilistic independence and union bounds.
Future work includes proving better accuracy bounds for differentially private
algorithms, and exploring further integration of \SYSTEM into \EasyCrypt.

\paragraph*{Acknowledgments.}
We thank the reviewers for their helpful comments. This work benefited from
discussions with Dexter Kozen, Annabelle McIver, and Carroll Morgan. This work
was partially supported by ERC Grant \#679127, and NSF grants 1513694 and
1718220.

\bibliographystyle{splncs03}
\bibliography{header,phl}

\clearpage

\iffull

\appendix

\section{Soundness of \SYSTEM}\label{app:soundness}

Before presenting the proof of soundness, we will introduce two technical lemmas
needed for the loop rules.  Intuitively, $d$- and $t$-closed assertions are
preserved in the limit of general and lossless loops, respectively.

\begin{prop}\strut \label{lem:close}
  \begin{enumerate}
  \item If $\form$ is $d$-closed and is s.t.
    \begin{gather*}
      \forall \mu .\, \mu \models \form \implies
        \dsem{\mu}{\ift{b}{s}} \models \form,
    \end{gather*}
    then $\forall \mu .\, \mu \models \form
      \implies \dsem{\mu}{\while{b}{s}} \models \form$.

  \item If $\form$ is $t$-closed and is s.t.
    \begin{gather*}
      \forall \mu .\, \mu \models \form \implies
        \dsem{\mu}{\ift{b}{s}} \models \form,
    \end{gather*}
    then $\forall \mu .\, \mu \models \form
      \implies \dsem{\mu}{\while{b}{s}} \models \form$,
    provided that $\while{b}{s}$ is lossless.
  \end{enumerate}
\end{prop}

\begin{proof}
  We only treat the first case; the second is similar.  Let $\form$ be a
  \kdclosed assertion s.t. for any sub-distribution $\mu$, if $\form(\mu)$, then
  $\dsem{\mu}{\ift{b}{s}} \models \form$.
  We prove by induction on $n$ that for any sub-distribution $\mu$ such that
  $\mu \models \form$, we have
  $\dsem{\mu}{(\ift{b}{s})^n} \models \form$.
  By downward closedness of $\form$, we have
  $\dsem{\mu}{(\ift{b}{s})^n_{|\neg b}} \models \form$.
  We conclude by $t$-closenedness of $\form$. \qedhere
\end{proof}

\begin{lem}\label{lem:semsound}
  Let $\form$ be an assertion and $s$ be a command. Then
  $\hoare {\form[\dsem{}{s}]} s {\form}$.
\end{lem}

\begin{proof}
  Let $\mu \models \form[\dsem{}{s}]$.
  By definition of $\form[\dsem{}{s}]$, this amounts to have
  $\dsem{\mu}{s} \models \eta$, which exactly gives the expected
  result.
\end{proof}

\begin{lem}\label{lem:semanti}
  Let $\mu$ be a sub-distribution and $s$ be a command.
  Then, $\wt{\dsem{\mu}{s}} \le \wt{\mu}$.
\end{lem}

\begin{proof}
  We have
  \begin{align*}
    \wt{\dsem{\mu}{s}}
      &= \wt{\dlet m \mu {\dsem{m}{s}}}
       = \sum_m \mu(m) \cdot \overbrace{\wt{\dsem{m}{s}}}^{\le 1} \\
      &\le \sum_m \mu(m) = \wt{\mu} \qedhere
  \end{align*}
\end{proof}

\begin{lem}\label{lem:lin}
  Let $\mu$ be a sub-distribution s.t.
  $\mu \eqdef \sum_{i \in I} \lambda_i \mu_i$ where $I$ is a
  finite set, all the $\lambda_i$'s are in $[0, 1]$ and all the
  $\mu_i$'s are sub-distributions.
  Then, for any command $s$,
  \begin{gather*}
    \dsem{\mu}{s} = \sum_{i \in I} \lambda_i \cdot \dsem{\mu_i}{s}
  \end{gather*}
\end{lem}

\begin{proof}
  Immediate consequence of the linearity of $\Exp$.
\end{proof}

\begin{lem}\label{lem:semnotmod}
  Let $s$ be a lossless command and $\mu$ be a sub-distribution. Then
  $\mu_{| \overline{\MV(s)}} = {\dsem{\mu}{s}}_{| \overline{\MV(s)}}$.
\end{lem}

\begin{proof}
  The proof is done by a direct induction of $s$.
\end{proof}

\begin{lem}
  The rules of \SYSTEM are sound.
\end{lem}

\begin{proof}
  We use the notation of the rules.

  \begin{itemize}
  \item \rname{Skip} --- immediate since $\dsem{\mu}{\skp} = \mu$.

  \item \rname{Abort} --- immediate by definition of $\detm{\form}$
    and since
    \begin{gather*}
    \supp(\dsem{\mu}{\abort}) = \supp(\dnull) = \emptyset.
    \end{gather*}

  \item \rname{Assgn} \& \rname{Sample} --- immediate consequence of
    \cref{lem:semsound}.

  \item \rname{Seq} --- let $\mu \models \form_1$. Then, by the first
    premise, $\dsem{\mu}{s_1} \models \form_2$. Hence, by the second
    premise, $\dsem{\dsem{\mu}{s_1}}{s_2} \models \form_3$.

  \item \rname{Conseq} --- if $\mu \models \form_0$, then
    $\mu \models \form_1$ by the first premise. Hence,
    $\dsem{\mu}{s} \models \form_2$ by the second premise and
    $\dsem{\mu}{s} \models \form_3$ by the third premise.

  \item \rname{Split} --- let $\mu \models \form_1 \oplus \form_2$.
    Then, there exists $\mu_1, \mu_2$ s.t. $\mu = \mu_1 + \mu_2$ and
    $\mu_i \models \form_i$ for $i \in \{1, 2\}$. By the two premises
    of the rule, we have $\dsem{\mu_i}{s} \models \form'_i$ for
    $i \in \{1, 2\}$.
    Now, by \cref{lem:lin}, we have
    $\dsem{\mu}{s} = \dsem{\mu_1}{s} + \dsem{\mu_2}{s}$.
    By taking resp.\ $\dsem{\mu_1}{s}$ and $\dsem{\mu_2}{s}$ for the
    witnesses of $\form'_1 \oplus \form'_2$, we obtain that
    $\dsem{\mu}{s} \models \form'_1 \oplus \form'_2$.

  \item \rname{Cond} --- we first prove that
  \begin{gather}
    \hoare {\form_1 \land \detm{e}} {\ifstmt e {s_1} {s_2}} {\form'_1}
    \label{eq:condT} .
  \end{gather}

  Let $\mu \models \form_1 \land \detm{e}$. Then, for $m \in \supp(\mu)$,
  we know that $\dsem{m}{e} = \True$. It follows that
  \begin{align*}
    \dsem{\mu}{\ifstmt e {s_1} {s_2}}
      &= \dlet m \mu
           {\underbrace{\dsem{m}{\ifstmt e {s_1} {s_2}}}_{= \dsem{m}{s_1}}} \\
      &= \dsem{\mu}{s_1}.
  \end{align*}

  However, by the first premise,
  $\dsem{\mu}{s_1} \models \form'_1$.
  Hence,
  \begin{gather*}
    \dsem{\mu}{\ifstmt e {s_1} {s_2}} \models \form'_1,
  \end{gather*}
  concluding the proof of \cref{eq:condT}.
  By a similar reasoning, we have:
  \begin{gather}
    \hoare {\form_2 \land \detm{\neg e}} {\ifstmt e {s_1} {s_2}} {\form'_2}
    \label{eq:condF} .
  \end{gather}

  We conclude the proof of soundness of~\rname{Cond} from the one
  of~\rname{Split} applied to \cref{eq:condT} and \cref{eq:condF}.

  \item \rname{Call} --- immediate consequence of~\rname{Seq}
    and~\rname{Assgn}.

  \item \rname{Frame} --- let $\mu \models \form$.
    Let $\mu \models \form$. To show that
    $\dsem{\mu}{s} \models \form$, from the definition of
    $\aindep {\form} {\MV(s)}$, it is sufficient to show that
    $\wt{\mu} = \wt{\dsem{\mu}{s}}$ and
    ${\mu}_{| \overline{\MV(s)}} = {\dsem{\mu}{s}}_{| \overline{\MV(s)}}$.
    The first one is a consequence of the losslessness of $s$, the
    second one is a direct application of \cref{lem:semnotmod}.

  \item \rname{While} --- from the first premise, by induction on $n$
    and using \rname{Seq}, we know that for any $n$, the following
    holds:
    \begin{gather}
      \hoare{\form}{(\ift{b}{s})^n}{\form}. \label{eq:witer}
    \end{gather}
    Now, let $\mu \models \form$. First, form the definition of
    $\dsem{\mu}{\while b s}$, we know that
    \begin{gather*}
      \dsem{\mu}{\while b s} \models \detm{\neg b}.
    \end{gather*}
    It remains to prove $\dsem{\mu}{\while b s} \models \form$.
    In the case of certain termination, we know the existence of a $k$
    s.t.
    \begin{gather*}
      \dsem{}{\while{b}{s}} = \dsem{}{(\ifstmt{b}{s})^k}
    \end{gather*}
    and we conclude by \cref{eq:witer}.
    For the cases of almost surely termination and almost termination,
    we conclude for \cref{eq:witer} and \cref{lem:close} of the paper.

  \item \rname{Adv} --- the proof is done by induction on the body of
    the external procedure which is of the form
    $s_1; c^?_1; \cdots; s_n; c^?_n$
    where the $s_i$'s do not contain calls and the $c^?_i$'s are
    potential calls to the oracles.
    From the \rname{Frame} rule, we know that the $s_i$'s preserve the
    invariant $\form$ --- noting that the losslessness of the
    adversary implies the losslessness of the $s_i$'s.
    Likewise, from the last premise of the \rname{Adv} rule, we know
    that calls to the oracles also preserve the invariant.
    Hence, by multiple application of the \rname{Seq} rule, we obtain
    that the adversary body maintains $\form$.
    \qedhere
  \end{itemize}
\end{proof}

\section{Semantics of Assertions}\label{app:semass}
The semantics of assertions is given \cref{fig:asr-semantics}.

\begin{figure}
\label{semantics_assert}
  \begin{equation*}
    \begin{array}{lllllll}
      \mdenot{v}  &\eqdef & \rho(v) \\
      \mdenot{\ind{\lasr}}  &\eqdef & \ind{\mdenot{\lasr}} \\
      \mdenot{\dlet{v}{g}{\lexp}} & \eqdef &
       \dlet{w}{\mdenot{g}}{\gdenot{\lexp}{\rho\subst{v}{w}}{m}} \\

      \mdenot{o(\lexp)} &\eqdef & o(\mdenot{\lexp}) \\
      \\ \hline \\
      \mdenot{\lexp_1 \bowtie \lexp_2}  &\eqdef & \mdenot{\lexp_1} \bowtie \mdenot{\lexp_2} \\
      \mdenot{FO(\lasr)} &\eqdef & FO(\mdenot{\lasr}) \\
      \\ \hline \\
      \pdenot{\pr{\lexp}{}}  &\eqdef & \dlet{m}{\mu}{\mdenot{\lexp}} \\
      \pdenot{o(\pexp)} &\eqdef & o(\pdenot{\pexp}) \\
      \\ \hline \\
      \pdenot{\pexp_1 \bowtie \pexp_2}  &\eqdef & \pdenot{\pexp_1} \bowtie \pdenot{\pexp_2} \\
      \pdenot{\form_1 \oplus \form_2} & \eqdef & 
      \exists \mu_1, \mu_2.\,  \mu = \mu_1 + \mu_2 \wedge 
        \denot{\form_1}^\rho_{\mu_1} \wedge \denot{\form_2}^\rho_{\mu_2}\\
      \pdenot{FO(\pasr)} &\eqdef & FO(\pdenot{\pasr}) \\
    \end{array}
  \end{equation*}
  \caption{Semantics of assertions}
  \label{fig:asr-semantics}
\end{figure}

\section{Precondition Calculus}

\Cref{prem_wp_full} contains the full definition of the precondition
calculus.

\begin{figure}
  \begin{equation*}
    \begin{array}{l@{\,}l}
      \prem(s, o({\pexp}))
      &\eqdef
      o({\prem(s, \pexp)})
      \quad \text{where } o \in \Ops 
      \\
      \prem(\skp, \pr{\lexp}{})
      &\eqdef
      \pr{\lexp}{}
      \\
      \prem(s_1; s_2, \pr{\lexp}{})
      &\eqdef
      \prem(s_1,\prem(s_2, \pr{\lexp}{}))
      \\
      \prem(x \asn e, \pr{\lexp}{})
      &\eqdef
      \pr{\lexp}{} \subst{x}{e}
      \\
      \prem(x \rnd g, \pr{\lexp}{})
      & \eqdef
      \Samp{x}{g}(\pr{\lexp}{})
      \\
      \prem(\ifstmt{e}{s_1}{s_2}, \pr{\lexp}{})
      & \eqdef
      \cond{\prem({s_1}, \pr{\lexp}{})}{e} + 
      \cond{\prem({s_2}, \pr{\lexp}{})}{\neg e}
      \\
      \prem(\abort, \pr{\lexp}{})
      &\eqdef
      0
      \\[1em]
      \wpre(s, \pexp_1 \bowtie \pexp_2)
      &\eqdef
      \prem(s, \pexp_1) \bowtie \prem(s, \pexp_2)
      \\
      \wpre(s,FO(\form))
      &\eqdef
      FO(\wpre(s,\form))
    \end{array}
  \end{equation*} 
\caption{Precondition calculus}
\label{prem_wp_full}
\end{figure}

\section{Soundness of the Syntactic Rules}\label{sec:soundness-syntactic}

\subsection{Certain Termination}

\begin{prop}[Soundness of rule \rname{While-$\sidecond{CTerm}$}]
   \label{while_DP1}
  Let $\mu$ be a sub-distributions such that $\sidecond{CTerm}$ is
  valid. Then,
   \[
     \denot{\while{b}{s}}{\mu} \models \eta \land \detm{\neg b}.
   \]
 \end{prop}
 \begin{proof}
  
%   Let $\mu$ be a sub-distribution satisfying the precondition $\form$.
%%   By assumption $\form$ also satisfies $\exists \dot{y}. \;
%%   \detm{\lexp \leq \dot{y}}$. Let $\Mem_n$ the support of the
%%   distribution $ \mu_n = \denot{{(\ift{b}{s})}^n}_\mu$ for $n \in
%%   \NN$.

  Given $\mu$ satisfying 
   $\sidecond{CTerm}$, 
   we first claim that there exists a decreasing function 
   $f: \NN \rightarrow \NN$ such that 
   $\detm{(\lexp \leq f(n) \vee \neg b)}$ holds at each iteration of
   the loop.
   Indeed, we remark that
   the  statement $\exists k \;
   \detm{(\lexp \leq k)}$ holds at first iteration  
   by the precondition hypothesis. Let:
   \[
     f(n) = \left\{
       \begin{array}{ll} 
         k-n & \text{if } n \leq k \\ 
         0 & \text{otherwise} .
       \end{array} \right.
     \]
  Then unrolling the loop and using the semantical \rname{Seq} rule 
  ensure by induction the claimed domination. The termination of the loop 
  arises then naturally from the exit condition $\detm{(\tilde{e}=0 \Rightarrow \neg b)}$.
 \end{proof}

 \subsection{Almost-Sure Termination}

 The main challenge for proving the soundness of the  is proving termination; from there, we can
 conclude by rule \rname{While-$\sidecond{ASTerm}$}. Our arguments use basic notations and
 theorems from the theory of Markov processes.
 %%a full introduction to Markov
 %%theory is not feasible in this space. We collect some of the concepts we need
 %%here:
%% \begin{inparaitem}[-]
%%   \item Positive recurrent;
%%   \item Null recurrent;
%%   \item Transient;
%%   \item Stationary distribution;
%%   \item Coupling;
%%   \item Stochastic dominance;
%%   \item Lifted Markov chain;
%%   \item Absorbing state.
%% \end{inparaitem}
%% These concepts all belong to the basic theory of Markov processes.

 \begin{prop}[Soundness of rule \rname{While-$\sidecond{ASTerm}$}]
   Let $\form$ be a $t$-closed assertion. For any sub-distributions
   $\mu$, such that $\sidecond{ASTerm}$ is valid.
   Then 
   \[
     \denot{\while{b}{s}}{\mu} \models \form .
   \]
 \end{prop}
 \begin{proof}
   As indicated, by soundness of the semantic while rule for almost-sure
   termination, and the premises, it suffices to prove almost-sure termination. 
      The sketch of the proof is the following:
      \begin{enumerate} 
        \item We follow the behavior of the variant by seeing it 
          as a random variable on the space of state.
        \item We first introduce a 
          Markov chain that reaches a particular state (the state zero) 
          with probability $1$.
        \item We then stochasticaly dominates 
          the variant by the latter chain.  
        \item In particular, this shows that the probability of the variant 
          to eventually reach $0$ is $1$, ensuring almost-sure termination.
      \end{enumerate}

   \textbf{Step one: Variant as a Random variable}. \\
   We consider the integer-valued variant $\lexp$ as a random variable 
   over the space of states. Let the family ${(\lexp_i)}_i$ represents the value taken by 
   $\lexp$ after the $i$-th iteration of the loop. 
   Then we have, using the semantical $\rname{Seq}$ rule and the 
   \begin{itemize}
     \item The sequence is uniformly bounded by $K$.
     \item The probability of decreasing is bounded below by $\epsilon$:
       \[
         \forall i \in \NN.\; \Pr[\lexp_i > \lexp_{i+1} \mid \lexp_i, \ldots, \lexp_0] = 
         \epsilon_i \geq \epsilon > 0 .
       \]
   \end{itemize}

   \textbf{Step two: Modelization with a Markov chain}. \\
   First, we can assume that $K > 0$ since if $K = 0$, the loop terminates
   immediately.
   Consider the following finite Markov chain $(X_i)_i$, over the state 
   space $S=\{0,\ldots, K\}$, following the transitions:
%   \begin{equation*}
%%     P_{a, b} \eqdef \left\{
%%       \begin{array}{ll}
%%        \epsilon   & \text{if } b = a-1   \\
%%         1-\epsilon & \text{if } b = K     \\
%%         1   & \text{if } a = b = 0 \\
%%         0   & \text{otherwise}.
%%       \end{array}
%%     \right.
%%   \end{equation*}
  
   \begin{center}
     \includegraphics[scale=1]{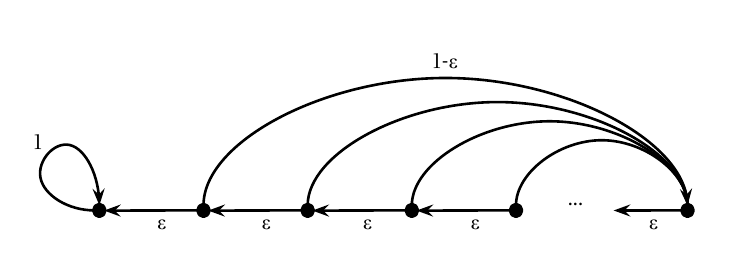}
   \end{center}

   This chain models the following behavior: with probability $\epsilon$ the value
  decrease by one, while with probability $1-\epsilon$ it jumps to $K$. Since zero is the only 
   connected component of the underlying graph, the 
   probability of terminating in the state zero is $1$ by a standard result on
   Markov chains.
   
   \textbf{Step three: Stochastic domination}. \\
   By construction, there exists a natural coupling between $(\tilde{X}_i)$ and
   $(\lexp_i)_i$, since the probability of the event \emph{$\lexp_i$ decrease}
   is greater than the probability of the event $\tilde{X}_i$ decrease.
  Since we impose that $X_0 = \lexp_0$ 
   (initial position), a simple induction over $i$ ensure that we have, for this coupling: 
   $\Pr[\lexp_i > 0] \leq \Pr[X_i > 0]$, so
   $\lim_{i \to \infty} \Pr[\lexp_i = 0] \geq 1 - \lim_{i \to \infty} \Pr[X_i > 0] = 1$ as desired.

 \end{proof}

\section{Further Examples}\label{sec:further-examples}

\subsection{Randomization for Approximation: Vertex Cover}

We begin with a classical application of randomization: \emph{approximation
  algorithms} for computationally hard problems. For problems that take long
time to solve in the worst case, we can sometimes devise efficient algorithms
that find a solution that is ``nearly'' as good as the true solution.

Our first example illustrates a famous approximation algorithm for the
\emph{vertex cover} problem. The input is a graph described by vertices $V$ and
edges $E$. The goal is to output a vertex cover: a subset $C \subseteq V$ such
that each edge has at least one endpoint in $C$, and such that $C$ is as small
as possible.

It is known that this problem is NP-complete, but there is simple
randomized algorithm that returns a vertex cover that is at on average
at most twice the size of the optimal vertex cover.  The algorithm
proceeds by maintaining a current cover (initially empty) and
considering each edge in order. If neither endpoint is in current
cover, the algorithm adds one of the two endpoints uniformly at
random. The \SYSTEM program is:

\begin{lstlisting}
proc VC (E : set<edge>) :
var   set<node> C  :=  $\emptyset$;
for (e$\mbox{}_1$,e$\mbox{}_2$) in E do
  if (e$\mbox{}_1$ $\notin$ C) $\land$ (e$\mbox{}_2$ $\notin$ C) then
    b ~~ $\{$0,1$\}$;
    C  :=  (b ? e$_1$ : e$_2$) $\cup$ C;
  fi
end
\end{lstlisting}

Here, we represent edges as a finite set of pairs of nodes. We loop
through the edges, adding one point of each uncovered edge to the
cover $\lstt{C}$ uniformly at random. The operator
$\lstt{b}\; ? \; \lstt{e}_1 : \lstt{e}_2$ returns $\lstt{e}_1$ if
$\lstt{b}$ is true, and $\lstt{e}_2$ if not.

To prove the approximation guarantee, we first assume that we have a
set of nodes $\lstt{C}^*$. We only assume that $\lstt{C}^*$ is a valid
vertex cover; i.e., each edge has at least one endpoint in
$\lstt{C}^*$. Then, we use the following loop invariant:

\begin{equation} \label{eq:vc-inv}
  \ex{\lstt{size} (\lstt{C} \setminus \lstt{C}^*)}
  \leq \ex{\lstt{size} (\lstt{C} \cap \lstt{C}^*)} .
\end{equation}

Given the loop invariant, we can prove the conclusion by letting
$\lstt{C}^*$ be the cover $OPT$ of minimal size, and reasoning about
intersections and differences of sets.

Clearly the invariant is initially true. To see why the invariant is
preserved, let $\lstt{e}$ be the current edge, with both endpoints out
of $\lstt{C}$. Since $\lstt{C}^*$ is a vertex cover, it has at least
one endpoint of $\lstt{e}$. Since our algorithm includes an endpoint
of $\lstt{e}$ uniformly at random, the probability we choose a vertex
not in $\lstt{C}^*$ is at most $1/2$, so the expectation on the left
in \cref{eq:vc-inv} increases by at most $1/2$. If $\lstt{e}$ is not
covered in $\lstt{C}$ but is covered by $\lstt{C}^*$, there is at
least a $1/2$ probability that we increase the intersection
$\lstt{C} \cap \lstt{C}^*$, so the right side in \cref{eq:vc-inv}
increases by at least $1/2$.  Thus, the invariant is preserved, and we
can prove
\begin{gather*}
  \hoare
  {\mathrm{isVC}(\lstt{C}^*, \lstt{E})}
  {\lstt{VC(E)}}
  {\ex{\lstt{size} (\lstt{C} \setminus \lstt{C}^*)}
  \leq \ex{\lstt{size} (\lstt{C} \cap \lstt{C}^*)}}
\end{gather*}
and by reasoning on intersection and difference of sets,
\begin{gather*}
  \hoare
  {\mathrm{isVC}(\lstt{C}^*, \lstt{E}) }
  {\lstt{VC(E)}}
  {\ex{\lstt{size} (\lstt{C})}
  \leq 2 \cdot \ex{\lstt{size} (\lstt{C}^*)}} .
\end{gather*}

\subsection{Random Walks: Termination and Reachability}

A canonical example of a random process is a \emph{random walk}. There
are many variations, but the basic scenario describes a numeric
position changing over time, where the position depends on the
position at the previous timestep, influenced by random noise drawn
from some distribution.

To demonstrate how to verify interesting facts about random processes,
we will model a one-dimensional random walk on the natural numbers.
We start at position $0$, and repeatedly update our position according
to the following rules. From $0$, we always make a step to $1$. From
non-zero positions, we flip a fair coin that is biased to come up
heads with probability $p > 0$. If heads, we increase our position by
$1$; if tails, we decrease by $1$.

We will prove two facts about this random walk. First, for any natural
number $T$, the probability of eventually reaching $T$ is $1$. Second,
when we reach $T$, we must first pass through all intermediate points
$1, \dots, T - 1$.  In \SYSTEM, we can express the random walk with
the following code.

\begin{lstlisting}
var bool visited[T];
var int  pos = 0;
visited[0]  :=  true;
while pos <> T do:
  c   ~~ Ber(p);
  pos  :=  pos + ((pos = 0) $\lor$ c) ? 1 : -1;
  visited[pos]  :=  true;
end
\end{lstlisting}

In order to verify this example, we will use the probabilistic while
rule \rname{While-ASTerm}.  First we establish almost-sure termination by
finding an appropriate termination measure: the distance between our
current position, and $\lstt{T}$. Indeed, this measure is bounded by
$\lstt{T}$, and has a non-negative (at least $1/2$) probability of
decreasing each iteration. Therefore, the loop terminates
almost-surely, and thus our random walk eventually reaches any point
$\lstt{T}$ with probability $1$.

To prove our second assertion---that we visit every point from $0$ to $\lstt{T}$---we
use the following loop invariant for the while command:
\begin{gather*}
  \detm{(\forall i.\; 0 \leq i \leq \lstt{pos} \Rightarrow \lstt{visited[i]} =
  \lstt{true})} .
\end{gather*}

In other words, if we have reached position $\lstt{pos}$, then we must
have already reached every position in $[0, \lstt{pos}]$. Since this
invariant is $t$-closed, we may apply rule \rname{While-ASTerm} and the
invariant holds at the end of the loop.  With the assertion
$\lstt{pos} = \lstt{T}$ at termination, we have enough to prove that
each position is visited:
\begin{gather*}
  \hoare{\llass}
  {s}
  {\llass
    \land
    \detm{(\forall i \in [\lstt{T}].\; \lstt{visited[i]} =
    \lstt{true})}} .
\end{gather*}

The losslessness post-condition indicates that the walk terminates
almost-surely.

\subsection{Amplification: Polynomial Identity Testing}

A second use of randomness is in running independent trials of the
same algorithm. This technique, known as \emph{probability
  amplification}, runs a randomized algorithm several times in order
to reduce the error probability.  Roughly speaking, a single trial may
have high error with some probability, but by repeating the trial it
is unlikely that all of the trials yield poor results.  By combining
the results appropriately---e.g., with a majority vote for algorithms
with binary outputs, or by selecting the best answer when we can check
the quality of the solution---we can produce an output that is more
accurate than a single run of the original algorithm.

An example in this vein is probabilistic \emph{polynomial identity
  testing}.  Given two multivariate polynomials $P(x_1,\ldots , x_n)$
and $Q(x_1,\ldots,x_n)$ over the finite field $\FF_q$ of $q$ elements,
we want to check whether $P = Q$, or equivalently, whether the
polynomial $P-Q$ is zero or not.  We will take $n$ uniformly random
samples $(v_i)_i$ from $\FF_q$ and check whether
$(P - Q)(v_1, \dots, v_n) = 0$. We repeat the trial $q$ times,
rejecting if we see a sample where the difference is non-zero.  In our
system, this corresponds to the following program:

\begin{lstlisting}
var bool res = true;
for i = 1 to q do:
  v   ~~ Unif($\FF^n_q$);
  res  :=  res $\wedge$ ((P - Q)(v) = 0);
end
\end{lstlisting}

The proof uses an instance of the Schwartz-Zippel lemma due to
{\O}ystein Ore for finite fields, which upper bounds the probability
of randomly picking a root of a polynomial over a finite field.

\begin{lem}
  Let $P$ a non-zero polynomial function over $\FF_q$. If we sample the
  variables $v_1, \dots, v_n$ uniformly at random from $\FF_q$,
  then
 \[ \Pr[ P(v_1, \dots, v_n) = 0 ] \leq 1 - 1/q .  \]
\end{lem}

We encode this lemma as an axiom in our system:
\begin{gather*}
  P \neq 0 \land
  v \sim \unifD(\FF^n_q)
  \Rightarrow
  \Pr [ P(v) = 0 ] \leq 1 - 1/q .
\end{gather*}

With this fact, we can prove the following loop invariant:
\begin{gather*}
  \lstt{P} \neq \lstt{Q} \Rightarrow \Pr[\lstt{res} = \lstt{true}]
  \leq {(1-1/\lstt{q})}^{i} ,
\end{gather*}
finally proving that
\begin{gather*}
  \hoare{\lstt{P} \neq \lstt{Q} }
  {s}
  {\Pr[\lstt{res} = \lstt{true}]
  \leq {(1-1/\lstt{q})}^{\lstts{q}} \leq 1/e} .
\end{gather*}

We have also verified Freivald's algorithm, an amplification-based
algorithm for checking matrix multiplication.

\subsection{Concentration Bounds: Private Running Sums}

Now, we turn to examples involving independence of random
variables. Our first example is drawn from the \emph{differential
  privacy} literature. Given a list of $N$ integers, we add noise from
a two-sided geometric distribution to each entry in order to protect
privacy, and we calculate the \emph{partial sums} of the noisy values:
$x_1, x_1 + x_2, x_1 + x_2 + x_3$, and so on. We wish to measure how
far the noisy sums deviate from the true sums.

In \SYSTEM, we can express this algorithm as follows:
\begin{lstlisting}
  var int X[N], noise[N], out[N];
  var int acc = 0;
  for i = 1 to N do
    noise[i] ~~ twogeom($\epsilon$);
    out[i]   :=  acc + X[i] + noise[i];
    acc      :=  out[i];
  end
\end{lstlisting}
The parameter $\epsilon$ to the noise distribution $\lstt{twogeom}$ is
a numeric parameter controlling the strength of the privacy guarantee,
by changing the magnitude of the noise.

Our loop invariant tracks three pieces of information: (i)
$\lstt{noise[i]}$ is distributed according to
$\lstt{twogeom}(\epsilon)$; (ii) the array $\lstt{noise}$ remains
independent; and (iii) $\lstt{out[i]}$ stores the noisy running sum:
\begin{gather*}
  \forall q \in [\lstt{i}].\; \lstt{out[q]} - \sum_{p \in [q]} \lstt{X}[p] =
  \sum_{p \in [q]} \lstt{noise}[p] .
\end{gather*}

To bound the error introduced by the noise, we need to bound
$\left| \sum_{p \in [q]} \lstt{noise}[p] \right|$. Since we know that
the elements in $\lstt{noise}$ are all independent, we can apply a
concentration bound to bound the probability of a large error in the
$p$-th running sum, concluding:
\begin{gather*}
  \{ \llass \} \;
  s \;
  \left\{ \forall p \in [\lstt{N}].\; \Pr \left[ \left| \textstyle\sum_{i \in [p]}
        \lstt{X}_i \right| \geq T \right] \leq Q(\epsilon, T)/\sqrt{\lstt{N}}
  \right\}
\end{gather*}
for a particular function $Q$ derived from the Berry-Esseen theorem.

\subsection{Hypercube Routing in More Details}

We will begin with the \emph{hypercube routing} algorithm
\citep{valiant1982scheme,Valiant:1981:USP:800076.802479}.  To set the stage,
consider a network where each node is labeled by a bitstring of length $D$, and
two nodes are connected by an edge if and only if the two corresponding labels
differ in exactly one bit position. This network topology is known as a
\emph{hypercube}, a $D$-dimensional version of the standard cube; a simple
example with $D = 3$ is in \cref{fig:hypercube}.  
% \begin{wrapfigure}{l}{0.2\textwidth}
\begin{figure}
  \begin{center}
  \includegraphics[width=0.6\textwidth]{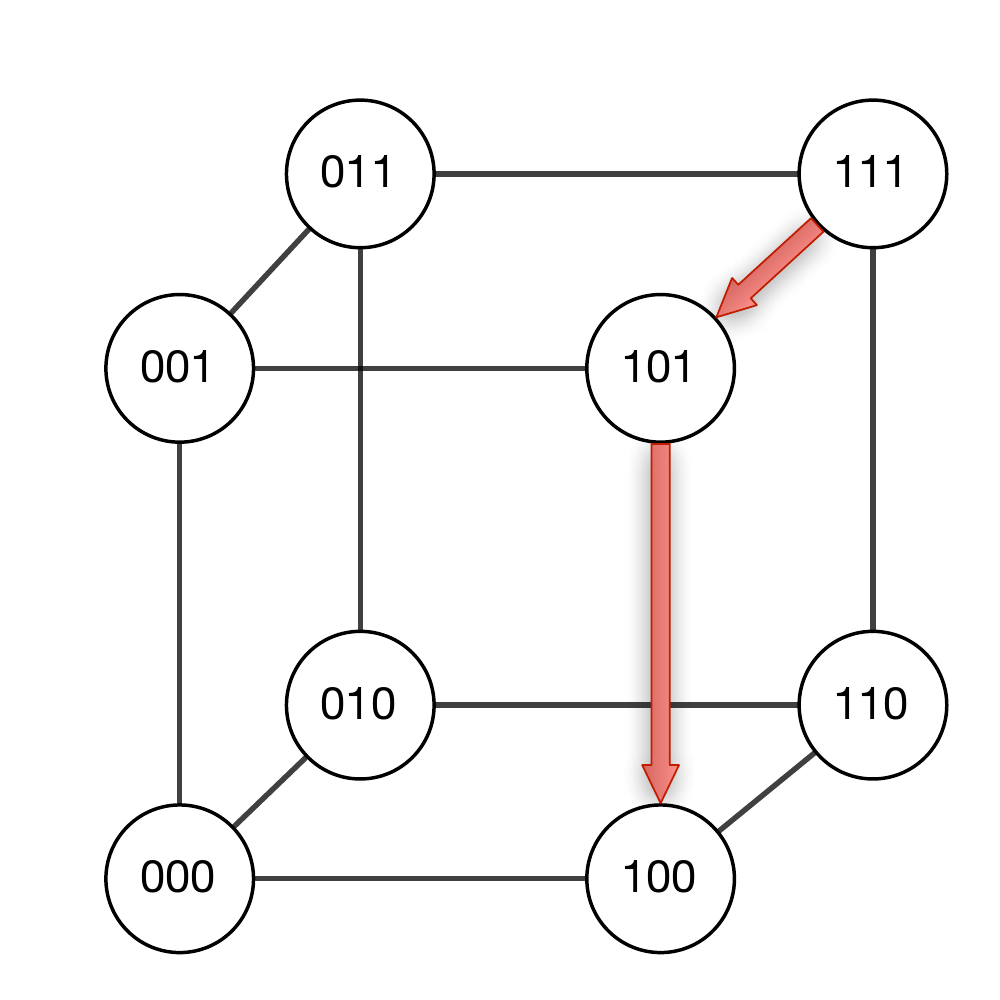}
  \caption{Hypercube path from $111$ to $100$ ($D = 3$)}
  \label{fig:hypercube}
  \end{center}
\end{figure}
% \end{wrapfigure}
In the network, there is
initially one packet at each node, and each packet has a unique destination. Our
goal is to design a routing strategy that will move the packets from node to
node, following the edges, until all packets reach their destination.
Furthermore, the routing should be \emph{oblivious}: to avoid coordination
overhead, each packet must select a path without considering the behavior of
the other packets.

To model the flow of packets, each packet's current position is a node in the
graph. Time proceeds in a series of steps, and at each step at most one packet
can traverse any single edge. If several packets try to use the same edge
simultaneously, one packet will be selected to move (for any selection strategy
that selects \emph{some} packet). The other packets wait at their current
position; these packets are \emph{delayed} and make no progress this step.

The routing strategy is based on \emph{bit fixing}: if the current position has
bitstring $i$, and the target node has bitstring $j$, we compare the bits in $i$
and $j$ from left to right, moving along the edge that corrects the first
differing bit. For instance, if we are at $111$ and we wish to reach $100$,
we will move along the edge corresponding to the second position: from $111$ to
$101$, and then along the edge corresponding to the third position: from $101$
to $100$. See \cref{fig:hypercube} for a picture.

While this strategy is simple and oblivious, there are permutations $\pi$ that
require a total number of steps growing linearly in the number of packets to
route all packets. Valiant proposes a simple modification, so that the total
number of steps grows \emph{logarithmically} in the number of packets.  In the
first phase, each packet $i$ will first select an intermediate destination
$\rho(i)$ uniformly at random from all nodes, and use bit fixing to reach
$\rho(i)$. In the second phase, each packet will use bit fixing to go from
$\rho(i)$ to the destination $\pi(i)$.  We will focus on the first phase,
since the reasoning for the second phase is nearly identical.  We can model the
strategy with the following code, using some syntactic sugar for the $\kfor$
loops; recall that the number of node in a hypercube of dimension $D$ is
$2^D$ so each node can be identified by a number in $[1,2^D]$.
\begin{lstlisting}
proc route ($D$ $T$ : int) :
  var $\rho$, pos, usedBy : node map;
  var nextE : edge;
  pos  :=  Map.init id $2^D$; $\rho$  :=  Map.empty;
  for $i$  :=  1 to $2^D$ do $\rho$[i] ~~ $[1,2^D]$
  for $t$  :=  1 to $T$ do
    usedBy  :=  Map.empty;
    for $i$  :=  1 to $2^D$ do
      if pos$[i] \neq \rho[i]$ then
        nextE  :=  getEdge pos[$i$] $\rho[i]$;
        if usedBy[nextE] = $\bot$ then
          usedBy[nextE]  :=  $i$; // Mark edge used
          pos[$i$]  :=  dest nextE // Move packet
  return (pos, $\rho$)
\end{lstlisting}
We assume that initially the position of the packet $i$ is at node $i$
(see \lstt{Map.init}). Then, we initialize the random intermediate 
destinations $\rho$. The remaining loop encodes the
evaluation of the routing strategy iterated $T$ time.
The variable \lstt{usedBy} is a map that logs if an edge is already
used by a packet, it is empty at the beginning of each iteration.
For each packet, we try to move it across one edge along the path to its
intermediate destination.
The function \lstt{getEdge} returns the next edge to follow, following the
bit-fixing scheme. 
If the packet can progress (its edge is not used), then its current 
position is updated and the edge is marked as used.

Our goal is to show that if the number of timesteps $T$ is $4 D + 1$, then all
packets reach their intermediate destination in at most $T$ steps, except
with a small probability $2^{-2D}$ of failure. That is, the number of
timesteps grows linearly in $D$, logarithmic in the number of packets.
At a high level, the analysis involves three steps.

% \subparagraph*{Average Path Load: Reasoning about Expectation.}

The first step is to consider a single packet traveling from $i$ to $\rho(i)$,
and to calculate the average number of other packets that share at least one
edge with $i$'s path $P$.  Let $H^\rho(i, j)$ be $1$ if the paths of packets $i$ and
$j$ share at least one edge and $0$ otherwise. Then, the load $R^\rho(i)$ on 
$i$'s path is the sum of $H^\rho(i,j)$ over all the other packets $j$.
%% In \SYSTEM, we can define these quantities as follows:
%% \begin{lstlisting}
%%  op H (i j:node) ($\rho$:node array) =
%%    $\ind{\texttt{i} \neq \texttt{j}} \land$ cross (mk $i$ $\rho[i]$) (mk $j$ $\rho[j]$))
%%  op R ($i$:node) ($\rho$:node array) =
%%    $\Sigma$ [1, 2^D] (fun j =>$\;$ H i j $\rho$)
%% \end{lstlisting}
%
By using the fact that each intermediate destination in $\rho$ is
uniformly distributed, and by analyzing the number of packets beside $i$'s
that use each edge on $i$'s path, we can upper bound the expected value of
$R^\rho(i)$ by $D/2$.
%
% \subparagraph*{High-Probability Bound: Independence and Concentration.}
In the second step, we move from a bound on the expectation of $R^\rho(i)$ to a
\emph{high-probability bound}: we want to show that $R^\rho(i) < 3D$ holds except
with a small probability of failure. The key tool is the \emph{Chernoff bound},
which gives high probability bounds of sums of independent samples.
The libraries in \SYSTEM include a mechanized proof of a generalized Chernoff
bound, with the following statement (leaving off some of the parameters):

\begin{lstlisting}
lemma Chernoff $n$($d$:mem distr)($X$:int    ->   mem -> bool):
  let $X m = \Sigma_{i\in[1, n]}~X~i~m$ in
  E$[X | d] \leq \mu \land \lstt{indep}~[0, n]~(X~i)~d$
  => Pr$[X > (1 + \delta)\mu | d] \leq (e^\delta / (1 + \delta)^{1 + \delta})^\mu$
\end{lstlisting}

Returning to the proof, we bounded the expectation of $R^\rho$ in the previous step.
To apply the Chernoff bound, we need to show that for any packet $i$, the
expressions $H^\rho(i, j)$ are independent for all $j$.
This is not exactly true, since $H^\rho(i, j_1)$ and $H^\rho(i, j_2)$ 
both depend on the (random) destination
$\rho(i)$ of packet $i$. However, it suffices to show that these variables are
independent if we fix the value of $\rho(i)$; then we can apply the Chernoff
bound to upper bound $R^\rho$ with high probability.

% \subparagraph*{Bounding the Delay.}

Finally, we can bound the total delay of any packet. This portion of the proof
rests on an intricate loop invariant assigning an imaginary coin for each delay
step to some packet that crosses $P$.  By showing that each packet holds at most
one coin, we can conclude that the $i$'s delay is at most the number $R^\rho(i)$ of
crossing packets. With our high probability bound from the previous step, we
can show that $T = 4D + 1$ timesteps is sufficient to route all packets $i$ to
$\rho(i)$, except with some small probability:
\[
  \hoare{T = 4D + 1}{\lstt{route}}{\Pr[ \exists i.\; \lstt{pos}[i] \neq
    \lstt{$\rho$}[i] ] \leq 2^{-2D}]}
\]

\subsection{Coupon Collector in More Details}
Our second example is the \emph{coupon collector} process. The
algorithm draws a uniformly random coupon (we have $N$ coupon) on each day, 
terminating when it has drawn at least one of each kind of coupon. 
The code of the algorithm is as follows.
\begin{lstlisting}
proc coupon ($N$ : int) :
  var int cp[$N$], $t$[$N$];
  var int $X$  :=  0;
  for $p$  :=  1 to $N$ do
    ct  :=  0;
    cur ~~ $[1,N]$;
    while cp[cur] = 1 do
      ct  :=  ct + 1;
      cur ~~ $[1,N]$;
    $t$[$p$]  :=  ct;
    cp[cur] :=  1;
    $X$  :=  $X$ + $t$[$p$];
  return $X$
\end{lstlisting}
The array \lstt{cp} keeps track of the coupons seen so far; initially
$\lstt{cp}[i]=0$.  We divide the loop into a sequence of phases (the outer loop)
where in each phase we repeatedly sample coupons and wait until we see a new
coupon (the inner loop). We keep track of the number of steps we spend in each
phase $p$ in $\lstt{t}[p]$, and the total number of steps in $X$.

Our goal is to bound the average number of iterations. The code involves two
nested loops, so we have two loop invariants. The inner loop has a probabilistic
guard: at every iteration, there is a finite probability that the loop
terminates (i.e., if we draw a new coupon), but there is no finite bound on the
number of iterations we need to run.
Since our desired loop invariant is not downwards closed, we must apply the rule
for almost sure termination. We use a variant that is $1$ if we have not seen a
new coupon, and $0$ if we have seen a new coupon. Note that each iteration, we
have a strictly positive probability $\rho(p)$ of seeing a new coupon and
decreasing the variant. Furthermore, the variant is bounded by $1$, and the loop
exits when the variant reaches $0$. So, the side-condition $\sidecond{ASTerm}$
holds.
For the inner loop, we prove that forall $c$
the invariant $\form_{i}$ is preserved:
\[
  \bigvee
  \begin{cases}
    (\detm{(\lstt{cp[cur]} = 1 \Rightarrow c \leq \lstt{ct})}
    \mathrel{\land} \Pr[ \lstt{cp[cur]} = 0 \land c = \lstt{ct} ] = (1 - \rho(p))^c \rho(p)) \\
    \exists k \in [1, c).\;
    \begin{cases}
      \detm{(\lstt{cp[cur]} = 1 \Rightarrow \lstt{ct} = k)
      \land (\lstt{cp[cur]} = 0 \Rightarrow \lstt{ct} \leq k)} \\
      \Pr[ \lstt{cp[cur]} = 1 \land \lstt{ct} = k] = (1 - \rho(p))^k) .
    \end{cases}
  \end{cases}
\]
Note that this is a $t$-closed formula; there is an existential in the
second disjunction, but it has finite domain (for fixed $c$).

For intuition, for every natural number $c$ there are two cases:
Either we have already unrolled more than $c$ iterations, or not. The
first disjunction corresponds to the first case, since loops where the
guard is true all have $\lstt{ct} \geq c$, and the probability of
stopping at $c$ iterations is $(1 - \rho(p))^c
\rho(p)$---we see $c$ old coupons, and then a new one.
Otherwise, we have the second disjunction. The integer $k$ represents
the current number of unfoldings of the loop. If the loop is
continuing then $k = \lstt{ct}$. If the loop is terminated, it terminated
before the current iteration: $\lstt{ct} < k$. Furthermore, the
probability of continuing at iteration $k$ is $(1 -
\rho(p))^k$.
At the end of the loop we have $\detm{\lstt{cp[cur]} = 0}$.
So, by the first conjunct and some manipulations,
\[
\forall c \in \NN.\; \Pr[ c = \lstt{ct} ] =
(1 - \rho(p))^c \rho(p)
\]
holds when the inner loop exits, precisely describing the distribution
of iterations $\lstt{ct}$ as $\geomD(\rho(p))$ by definition.

The outer loop is easier to handle, since the loop has a fixed bound
$N$ on the number of iterations so we can use the rule for
certain termination.  For the loop invariant, we take:
\[
  \form_{out} \eqdef
  \left\{
  \begin{array}{l}
    \forall i \in [p - 1].\; t[i] \sim \lstt{Geom}(\rho(i))
    \mathrel{\land} \detm{\left(X = \sum_{i \in [p - 1]}
        t[i]\right)} \\
    \mathrel{\land} \detm{\left(\sum_{i \in [1,N]} \lstt{cp}[i] = p
        - 1\right)}\\
    \mathrel{\land} \forall i \in [1,N].\;
    \detm{(\lstt{cp}[i] \in \{ 0, 1 \})} .
  \end{array}
  \right.
\]

The first conjunct states that the previous waiting times follow a geometric
distribution with parameter $\rho(i)$; this assertion follows from the previous
reasoning on the inner loop.  The second conjunct asserts that $X$ holds
the total waiting time so far. The final two conjuncts state that there are at
most $p - 1$ flags set in \lstt{cp}.  Thus,
\[
  \{ \llass \}\;
  \lstt{coupon} \;
  \left\{
    \begin{array}{l}
      \forall i \in [1,N] . \; t[i] \sim \lstt{Geom}(\rho(i)) \\
      \mathrel{\land} \detm{X = \sum_{i \in [1,N]} t[i]}
    \end{array}
  \right\}
\]
at the end of the outer loop.  By applying linearity of expectations and a fact
about the expectation of the geometric distribution, we can bound the expected
running time:
\[
  \{ \llass \}\;
  \lstt{coupon} \;
  \left\{
      \ex{X} = \textstyle\sum_{i \in [1,N]} \left(
        \frac{N}{N - i + 1} \right)
  \right\} .
\]

\fi

\end{document}

%%% Local Variables:
%%% mode: latex
%%% TeX-master: t
%%% End:

%% file: prelude.tex
\usepackage[T1]{fontenc}
\usepackage[utf8]{inputenc} 
\usepackage{microtype}
\usepackage{paralist}
\usepackage{enumitem}
\usepackage{url}
\usepackage{stmaryrd}
\usepackage{graphicx}
\usepackage{xspace,xcolor}
\usepackage{booktabs}

\usepackage{amsthm}
\usepackage{amsmath,amsfonts,amssymb,mathtools}
\usepackage{nicefrac}
\usepackage{flushend}
\usepackage{mathpartir}
\usepackage{hyperref}
\usepackage[capitalize]{cleveref}
\usepackage{wrapfig}
\newcommand{\SYSTEM}{\mbox{\textsc{Ellora}}\xspace}

\newcommand{\EasyCrypt}{\textsc{EasyCrypt}\xspace}

\newcommand{\Sahl}{\textsc{aHL}\xspace}
\newcommand{\Spgcl}{\textsc{pGCL}\xspace}
\newcommand{\Sppdl}{\textsc{PPDL}\xspace}
\newcommand{\Sli}{\textsc{IL}\xspace}

\usepackage{todonotes}
\newcommand{\gb}[1]{\ifdraft\todo[author=Gilles, inline]{#1}\fi}

\newcommand{\jh}[1]{\ifdraft\todo[author=Justin, inline]{#1}\fi}

\newcommand{\rname}[1]{[\textsc{#1}]}

\let\xinferrule\inferrule
\renewcommand{\inferrule}[3][]{%
  \ifx&#1&
    \xinferrule*{#2}{#3}
  \else
    \xinferrule*[right=\rname{#1}]{#2}{#3}
  \fi}

\newcommand{\NN}{\mathbb{N}}
\newcommand{\NNinf}{\mathbb{N}^\infty}

\newcommand{\FF}{\mathbb{F}}

\newcommand{\Exp}{\mathbb{E}}
\newcommand{\ExpD}{\mathbb{E}}

\newcommand{\dnull}[1][]{{{\mathbb{0}}^{#1}}}
\newcommand{\dunit}[2][]{{{\delta}^{#1}_{#2}}}
\newcommand{\dlet}[3]{\ExpD_{{#1} \sim {#2}} [{#3}]}
\newcommand{\dslet}[2]{\ExpD_{#1} [{#2}]}
\newcommand{\drestr}[2]{{#1}_{| {#2}}}

\newcommand{\binD}{\mathrm{B}}
\newcommand{\detE}{\mathsf{det}}

\newcommand{\carac}[1]{\mathbf{1}_{#1}}

\newcommand{\pWhile}{\textsc{pWhile}\xspace}

\newcommand{\wt}[1]{|#1|}
\DeclareMathOperator{\supp}{supp}

\DeclareMathOperator{\MV}{mod}

\newcommand{\AProc}{\mathcal{A}}
\newcommand{\IProc}{\mathcal{I}}

\newcommand{\Vars}{\ensuremath{\mathrm{FV}}}

\newcommand{\Ops}{\ensuremath{\mathbf{Ops}}}
\newcommand{\Dist}{\ensuremath{\mathbf{SDist}}}

\newcommand{\pr}[2]{\dslet{#2}{#1}}
\newcommand{\ex}[1]{\pr{#1}{}}

\newcommand{\detm}[1]{\square{#1}}

\newcommand{\llass}{\mathcal{L}}
\newcommand{\cond}[2]{{#1}_{| #2}}

\newcommand{\True}{\ensuremath{\mathop{\top}}}

\newcommand{\form}{\pasr}
\newcommand{\lexp}{\tilde{e}}
\newcommand{\lasr}{\phi}

\newcommand{\pexp}{p}
\newcommand{\pasr}{\eta}

\newcommand{\state}{\mu}

\newcommand{\subst}[2]{[#1 := #2]}
\newcommand{\bernD}{\mathbf{Bern}}
\newcommand{\geomD}{\mathbf{Geom}}
\newcommand{\unifD}{\mathbf{Unif}}

\DeclareMathOperator{\indep}{\mbox{\bf \#}}

\newcommand{\Var}{\mathcal{X}}
\newcommand{\AVar}{\mathfrak{s}}

\newcommand{\VarL}[1][]{\mathcal{X}^{\mathfrak{L}}_{#1}}

\newcommand{\Mem}{\mathbf{State}}

\newcommand{\denot}[1]{\llbracket #1 \rrbracket}
\newcommand{\dsem}[2]{\denot{#2}_{#1}}

\newcommand{\gdenot}[3]{\denot{#1}^{#2}_{#3}}
\newcommand{\pdenot}[2][\mu]{\denot{#2}^{\rho}_{#1}}
\newcommand{\mdenot}[2][m]{\denot{#2}^{\rho}_{#1}}

\newcommand\bs[1]{\{0,1\}^{#1}}

\newcommand{\kw}[1]{\mathbf{#1}}

\newcommand{\kif}{\kw{if}}
\newcommand{\kthen}{\kw{then}}
\newcommand{\kelse}{\kw{else}}
\newcommand{\kwhile}{\kw{while}}
\newcommand{\kfor}{\kw{for}}
\newcommand{\kdo}{\kw{do}}

\newcommand\hoare[3]{\{#1\}\;#2\;\{#3\}}
\newcommand\ubhl[4]{\models_{#1} \{#2\}\;{#3}\;\{#4\}}

\newcommand{\indhoare}[3]{\left\{#1\right\}\;#2\;\left\{#3\right\}}
\newcommand{\follows}[2]{#1 \sim #2}

\newcommand{\rnd}{%
  \stackrel{\raisebox{-.15ex}[.25ex]{\tiny %
  $\mathdollar$}}{\raisebox{-.2ex}[.2ex]{$\leftarrow$}}}
\newcommand{\asn}{%
  \stackrel{}{\raisebox{-.1ex}[.2ex]{$\leftarrow$}}}

\newcommand{\ifstmt}[3]{\kif\: #1\: \kthen\: #2\: \kelse\: #3}
\newcommand{\ifte}[3]{\kif\: #1\: \kthen\: #2\: \kelse\: #3}
\newcommand{\ift}[2]{\kif\: #1\: \kthen\: #2}
\newcommand{\while}[2]{\kwhile\: #1 \: \kdo \: #2}

\newcommand{\call}[3]{{#1} \asn {#2}({#3})}

\newcommand{\ind}[1]{\mathbf{1}_{#1}}
\newcommand{\skp}{\kw{skip}}
\newcommand{\abort}{\kw{abort}}

\newcommand{\Samp}[2]{\mathcal{P}_{#1}^{#2}}
\DeclareMathOperator{\prem}{\text{pe}}
\DeclareMathOperator{\wpre}{\text{pc}}

\newcommand{\expr}{\mathcal{E}}
\newcommand{\dexpr}{\mathcal{D}}

\newcommand{\farg}[1]{{#1}_{\kw{arg}}}
\newcommand{\fbody}[1]{{#1}_{\kw{body}}}
\newcommand{\fret}[1]{{#1}_{\kw{res}}}

\newcommand{\aora}[1]{{#1}_{\kw{ocl}}}

\newcommand{\uclosed}[1]{\mathsf{uclosed}(#1)}
\newcommand{\tclosed}[1]{\mathsf{tclosed}(#1)}
\newcommand{\dclosed}[1]{\mathsf{dclosed}(#1)}
\newcommand{\sidecond}[1]{\mathcal{C}_{\mbox{{\small #1}}}}

\newcommand{\kdclosed}{\mbox{$d$-closed}\xspace}

\newcommand{\dsvalid}[2]{#1(#2)}

\newcommand{\eqdef}{\mathrel{\stackrel{\scriptscriptstyle \triangle}{=}}}

\newcommand{\assn}{\mathsf{Assn}}
\newcommand{\aindep}[2]{\mathsf{separated}({#1},{#2})}

\usepackage{listings}

\usepackage[scaled]{beramono}
\newcommand{\Small}{\fontsize{8.2pt}{8.4pt}\selectfont}
\newcommand*{\LSTfont}{\Small\ttfamily\SetTracking{encoding=*}{-60}\lsstyle}

\newcommand{\lstrnd}{\stackrel{\raisebox{-.15ex}{\ensuremath{\scriptscriptstyle\$}}}{\raisebox{-.2ex}{\ensuremath{\leftarrow}}}}

\newcommand{\lstt}[1]{\mbox{\LSTfont #1}}
\newcommand{\lstiny}{\fontsize{6.6pt}{6.8pt}\selectfont}
\newcommand{\lstts}[1]{\mbox{\lstiny\ttfamily\SetTracking{encoding=*}{-60}\lsstyle #1}}
 
\newtheorem{thm}{Theorem}
\newtheorem{lem}[thm]{Lemma}

\newtheorem{prop}[thm]{Proposition}

\crefname{section}{\S}{\S}
\Crefname{section}{\S}{\S}

\crefname{prop}{Proposition}{Propositions}
\Crefname{prop}{Proposition}{Propositions}

\crefname{lem}{Lemma}{Lemmas}
\Crefname{lem}{Lemma}{Lemmas}

\crefname{thm}{Theorem}{Theorems}
\Crefname{thm}{Theorem}{Theorems}